\documentclass[11pt]{article}

\usepackage[margin=1in]{geometry}
\usepackage{graphicx}
\usepackage[%
  hypertexnames=false,
  colorlinks=true,
  linkcolor=red,
  citecolor=black,
  urlcolor=blue%
]{hyperref}
\usepackage{authblk}
\usepackage[round, sort]{natbib}

\usepackage{amsmath, amssymb, amsthm}

\theoremstyle{plain}
\newtheorem{thm}{Theorem}
\newtheorem{lem}[thm]{Lemma}
\newtheorem{prop}[thm]{Proposition}

\theoremstyle{definition}
\newtheorem{defn}[thm]{Definition}

\newtheoremstyle{case}
  {.5\topsep}
  {.5\topsep}
  {\normalfont}
  {0pt}
  {\bfseries}
  {.}
  {1em}
  {\thmname{#1}\thmnumber{ #2}\textnormal{\thmnote{ (#3)}}}
\theoremstyle{case}
\newtheorem{case}{Case}

\newtheoremstyle{detail}
  {.25\topsep}
  {.25\topsep}
  {\normalfont}
  {\parindent}
  {\itshape}
  {.}
  {1em}
  {\thmnote{#3}}
\theoremstyle{detail}
\newtheorem{detail}{Detail}

\newcommand{\bb}[1]{\mathbf{#1}}
\newcommand{\B}[1]{\mathbb{#1}}
\newcommand{\C}[1]{\mathcal{#1}}
\newcommand{\D}[1]{\operatorname{\mathrm{#1}}}
\newcommand{\T}[1]{\mathrm{#1}}

\newcommand{\defeq}{\stackrel{\hbox{\tiny def}}{=}}

\title{A Simple, Statistically Robust Test of Discrimination\thanks{%
  We thank Chris Avery, Madison Coots, Eliana La Ferrara, Julian Nyarko, Todd
  Rogers, Soroush Saghafian, Kevin Yang, and Michael Zanger-Tischler for helpful
  conversations and feedback. Data and replication code are available at
  \url{https://github.com/jgaeb/outcomepp}.
}}

\author[1]{Johann D. Gaebler}
\author[2]{Sharad Goel}

\affil[1]{%
  Department of Statistics\protect\\%
  Harvard University, 1 Oxford Street, Cambridge, MA 02138, USA%
}
\affil[ ]{\relax}
\affil[2]{%
  Harvard Kennedy School of Government\protect\\%
  Harvard University, 79 John F. Kennedy Street, Cambridge, MA 02138, USA%
}

\date{}

\begin{document}

\maketitle

\begin{abstract}
  \noindent
  In observational studies of discrimination, the most common statistical
  approaches consider either the rate at which decisions are made (benchmark
  tests) or the success rate of those decisions (outcome tests). Both tests,
  however, have well-known statistical limitations, sometimes suggesting
  discrimination even when there is none. Despite the fallibility of the
  benchmark and outcome tests individually, here we prove a surprisingly strong
  statistical guarantee: under a common non-parametric assumption, at least one
  of the two tests must be correct; consequently, when \emph{both} tests agree,
  they are guaranteed to yield correct conclusions. We present empirical
  evidence that the underlying assumption holds approximately in several
  important domains, including lending, education, and criminal justice---and
  that our hybrid test is robust to the moderate violations of the assumption
  that we observe in practice. Applying this approach to 2.8 million police
  stops across California, we find evidence of widespread racial discrimination.
\end{abstract}

\thispagestyle{empty}

\section{Introduction}%
\label{sec:intro}

When assessing claims of discrimination, researchers often begin by considering
whether decision rates differ across groups defined by race or gender, typically
after adjusting for relevant differences between groups. For example, to test
for discrimination in banking, one might estimate differences in lending rates
between White and Black loan applicants after adjusting for an individual's
credit score, income, and savings. Although such a ``benchmark test'' can be
informative, it is prone to omitted-variable bias: failing to adjust for all
relevant information can yield misleading estimates. Nonetheless, benchmark
tests have been applied in nearly every domain where discrimination is studied,
generally under an implicit assumption that analysts have access to all relevant
covariates~\citep{%
  gelman2007analysis, macdonald2020effect, grossman2023disparate,
  starr2013mandatory, bartlett2022consumer, gaebler2022causal%
}.

To mitigate the omitted-variable problem inherent to benchmark tests, \citet{%
  becker1993nobel, becker1957economics%
}
introduced the ``outcome test,'' in which one looks not at decision rates but
rather \emph{success} rates. If, for example, loans issued to Black borrowers
are repaid at higher rates than those issued to White borrowers, it suggests a
double---and discriminatory---standard, with bank officials granting loans only
to exceptionally creditworthy Black applicants. Owing perhaps to its simplicity
and intuitive appeal, the outcome test has now become one of the most popular
empirical approaches to detecting discrimination. Researchers have applied the
test to audit a wide range of decisions, including lending, hiring, publication,
and candidate election~\citep{%
  berkovec1994race, berkovec1998discrimination, pope2011s,
  chilton2020political, smart1996citation, green2009gender, anzia2011jackie%
}.
The outcome test has gained particular prominence in criminal justice, among
both researchers and policymakers~\citep{%
  goel2016precinct, pierson2020large, neil2019methodological,
  coviello2015economic, antonovics2009new, fryer2019empirical,
  ridgeway2007analysis, persico2006generalising, goel2017combatting,
  close2007searching%
}.

\begin{figure}
  \begin{center}
    \includegraphics{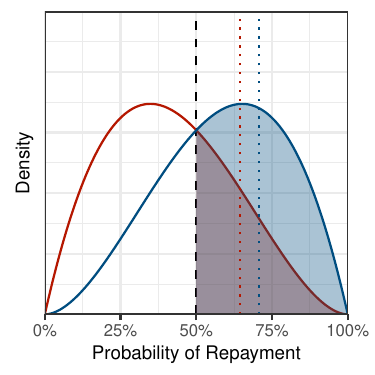}
  \end{center}
  \caption{\emph{%
    A stylized example illustrating the problem of inframarginality. The two
    curves depict the distribution of repayment probabilities for two
    hypothetical subpopulations. Applying a uniform lending threshold of 50\%
    (black vertical line) results in higher repayment rate for loan recipients
    in the blue group (71\%; blue vertical line) compared to recipients in the
    red group (64\%; red vertical line). The outcome test would thus incorrectly
    infer that members of the blue group were subjected to a more stringent
    lending standard.
  }}%
  \label{fig:inframarginality}
\end{figure}

Like the benchmark test, however, the outcome test suffers from well-known
statistical limitations~\citep{%
  ayres2002outcome, simoiu2017problem, carr1993federal, galster1993facts,
  engel2008critique, anwar2006alternative%
}.
Consider the stylized example in Figure~\ref{fig:inframarginality}, where the
red and blue curves show the distribution of repayment probability across loan
applicants in two different groups (henceforth, ``risk distributions''). In this
hypothetical, bank officials grant loans to those applicants who are at least
50\% likely to repay their loans---indicated by the black, dashed vertical
line---irrespective of group membership. Despite this uniform lending standard,
loan recipients in the blue group are more likely to repay their loans than
recipients in the red group. In statistical terms, conditional on being above
the lending threshold, the mean of the blue group is greater than the mean of
the red group. As a result, the outcome test would incorrectly conclude that
applicants in the blue group were subject to a more stringent lending standard.

This problem of ``inframarginality'' has attracted considerable attention,
prompting several attempts to place outcome tests on firmer statistical footing.
\citet{knowles2001racial} developed a model of behavior under which risk
distributions collapse to a single point, eliminating the possibility of
inframarginality. Although theoretically interesting, the key assumption in that
approach has been critiqued for being at odds with empirical
evidence~\citep{engel2008searching, jung2024omitted}.
\citet{anwar2006alternative} proposed a test based on decision and outcome rates
conditional on the race of both decision makers and those subject to those
decisions. Their method is guaranteed, under certain conditions, to produce
correct inferences, but it can only identify relative disparities between
decision makers from different race groups. Building on that work,
\citet{alesina2014test} proposed a test of racial bias in capital sentencing
based on the relative likelihood that decisions are overturned across
defendant-victim race pairs.\@ \citet{arnold2018racial} seek to sidestep
concerns of inframarginality by directly estimating outcomes for individuals at
the margin, leveraging quasi-random assignment of decision makers. Theirs is a
statistically compelling approach, but can only be applied in certain settings,
where decision makers are plausibly quasi-randomly assigned and analysts have
information on the actions of individual decision makers.
\citet{simoiu2017problem} and \citet{pierson18} aim to overcome inframarginality
by simultaneously estimating risk distributions and decision thresholds with a
parametric model. Their approach, however, is sensitive to the exact model form,
and, in particular, estimates are not identified by the data alone. Finally,
\citet{jung2024omitted} use detailed individual-level information on covariates
and outcomes to directly estimate group-specific risk distributions. The method
is effective when it can be used~\citep{%
  grossman2024reconciling, grossman2023racial, souto2024differences%
},
though the demanding data requirements limit the applicability of their
approach.

Despite the limitations of both the benchmark and outcome tests, here we show
that simply combining the two yields a \emph{robust} outcome test with
surprisingly strong statistical guarantees. In particular, if the group-specific
risk distributions satisfy the monotone likelihood ratio property
(MLRP)~\citep{karlin1956theory}, then either the benchmark test---without
adjusting for any covariates---or the standard outcome test must yield correct
conclusions. Thus, when \emph{both} the benchmark and outcome tests indicate
discrimination, that conclusion must be correct. The MLRP is a widely applied
assumption on signal distributions in information economics~\citep[e.g.,][]{%
  milgrom1981good, grossman1992analysis, athey2018value, persico2000information%
},
as well as in the outcome test literature~\citep[e.g.,][]{%
  anwar2006alternative, feigenberg2022would%
}.
We expect the MLRP to hold when it is similarly difficult to make accurate
decisions for members of each group (e.g., when the group-specific risk
distributions have similar variances, as in Figure~\ref{fig:inframarginality}).
Drawing on data from lending, education, and criminal justice, we present
empirical evidence that the MLRP is approximately satisfied in several important
domains. We further show that our hybrid test is robust to the moderate
violations of the MLRP that we observe in our data. Applying this approach to
2.8 million police stops across 56 law enforcement agencies in California, we
find evidence of pervasive discrimination in police searches of Black and
Hispanic individuals.

\section{Statistical Guarantees}%
\label{sec:simple}

In our running lending example, our robust outcome test suggests discrimination
against a group if two conditions hold simultaneously: (1) lending rates are
\emph{lower} for that group (the benchmark test), and (2) repayment rates among
loan recipients are \emph{higher} for the group (the standard outcome test). In
the stylized example depicted in Figure~\ref{fig:inframarginality}, loan
recipients in the blue group have higher repayment rates, satisfying the
standard outcome test; but members of the blue group are also more likely to
receive loans, failing the benchmark test. In this case, whereas the standard
outcome test incorrectly infers the blue group is held to a higher,
discriminatory lending standard, our robust outcome test correctly concludes
that there is insufficient evidence to support a claim of discrimination. We
next present formal conditions under which the robust outcome test is guaranteed
to produce correct results.

\subsection{Formal Setup}

Our formal setup follows the literature on analyzing outcome
tests~\citep[cf.][]{simoiu2017problem}. We imagine a population of individuals
belonging to one of two groups \(G \in \{0, 1\}\), indicating, for example,
their race or gender. Decision makers take a binary action \(D \in \{0, 1\}\)
for each individual, such as approving (\(D = 1\)) or denying (\(D = 0\)) an
individual's application for a loan. The decision maker is interested in some
binary outcome \(Y \in \{0, 1\}\), which, in our running example, corresponds to
loan repayment (e.g., \(Y = 1\) if the loan is repaid and \(Y = 0\) otherwise).
The decision maker does not know \(Y\) at decision time, but they can estimate
it based on the information \(X \in \C X\) then available to them about the
applicant. In particular, at the moment the decision is made, we assume they can
estimate the probability that \(Y = 1\) given the available information:
\[
  R \defeq \Pr(Y = 1 \mid X).
\]
In our running example, \(R\) is the decision maker's estimate of the
applicant's repayment probability. Moreover, the conditional distributions of
\(R\) by group correspond to the risk distributions in
Figure~\ref{fig:inframarginality}.

Finally, we assume that decision makers are \emph{rational}, meaning that,
within each group, their actions follow threshold rules. (Below, we relax this
assumption.) In particular, we assume they take action \(D = 1\) for individuals
in group \(G = g\) if, and only if, \(R\) exceeds some (possibly group-specific)
threshold \(t_g\):
\[
  D \defeq \begin{cases}
    1 & \text{if} \ G = g \ \text{and} \ t_g \leq R \\
    0 & \text{otherwise}.
  \end{cases}
\]
Following \citet{becker1957economics, becker1993nobel}, ``discrimination'' in
this setting corresponds to having different group-specific thresholds (i.e.,
\(t_0 \neq t_1\)), meaning decision makers apply a double standard. For
instance, in our lending example, \(t_1 > t_0\) would mean that decision makers
grant loans to members of group \(G = 1\) only if they are exceptionally
qualified---amounting to discrimination against that group.

With this setup, we now state a simplified version of our main technical result.

\begin{prop}%
\label{prop:simple}
  Suppose \(\Pr(G = 1 \mid R = r)\) is a monotonic function of \(r\), and that,
  for \(g \in \{0, 1\}\), the conditional distribution \(R \mid G = g\) has
  positive density on \((0, 1)\). Further assume that the decision thresholds
  are non-degenerate, i.e., \(0 \leq t_g < 1\) for \(g \in \{0, 1\}\). Now, if:
  \begin{enumerate}
    \item \(\Pr(D = 1 \mid G = 0) > \B \Pr(D =1 \mid G = 1)\), meaning the
      decision rate for group \(G = 1\) is lower than for group \(G = 0\); and
    \item \(\Pr(Y = 1 \mid D = 1, G = 0) < \Pr(Y = 1 \mid D = 1, G = 1)\),
      meaning the outcome rate for group \(G = 1\) is higher than for group \(G
      = 0\);
  \end{enumerate}
  then \(t_0 < t_1\).
\end{prop}

Proposition~\ref{prop:simple} shows that under the stated monotonicity
assumption---which, we show below, is equivalent to the standard MLRP---a group
with both lower decision rates and higher outcome rates is necessarily being
held to a higher threshold. For ease of exposition, we present this result for
threshold decision rules and binary outcomes, but a much more general version of
the result holds. Theorem~\ref{thm:general} in Section~\ref{sec:general} extends
Proposition~\ref{prop:simple} to cover real-valued outcomes (e.g., repayment
amounts rather than a binary repayment indicator) and quasi-rational decision
makers (e.g., with decisions following a logistic curve rather than a threshold
function). To illustrate the key ideas behind the general result, we give a
proof of this special case.

\subsection{Proof of Proposition~\ref{prop:simple}}

Proposition~\ref{prop:simple} cannot hold without \emph{some} hypothesis on the
risk distributions, and Figure~\ref{fig:exception} illustrates one possible
failure of the robust outcome test in our running lending example. Here, there
is no discrimination, because lending decisions are made according to a uniform
threshold. However, lending rates are lower and repayment rates higher for the
blue group. One intuitive way of capturing the issue is that it is easier for
loan officers to determine whether an applicant will default in the blue group
than the red group because the risk distribution of the blue group has higher
variance than that of the red group. The MLRP formalizes and generalizes this
intuition, capturing the key properties needed for the robust outcome test to be
correct.

\begin{figure}
    \begin{center}
      \includegraphics{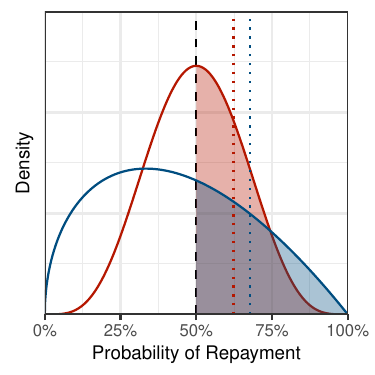}
    \end{center}
    \caption{\emph{%
      A stylized illustration of a pair of risk distributions and lending
      decisions for which the robust outcome test would be incorrect. The
      lending rates correspond to the areas of the colored regions (viz., 50\%
      for the red group and 38\% for the blue group), and the repayment rates
      are shown by the dotted red and blue lines (viz., 62\% for the red group
      and 67\% for the blue group). Applying a uniform lending threshold (50\%)
      results in a lower lending rate to applicants from the blue group, as well
      as a higher repayment rate, leading the robust outcome test to erroneously
      conclude that the thresholds differ for the two groups.
    }}%
\label{fig:exception}
\end{figure}

First, we note that the MLRP is equivalent to the---in our setting, more
intuitive---monotone conditional probability (MCP) condition in
Proposition~\ref{prop:simple}. Suppose that the conditional distributions of \(R
\mid G = g\) have positive densities given by \(f_{R \mid G = g}(r)\). The MLRP
simply states that the likelihood ratio
\[
  \frac {f_{R \mid G = 1}(r)} {f_{R \mid G = 0}(r)}
\]
is a monotonic function of \(r\). Recalling the monotonicity condition in
Proposition~\ref{prop:simple}, observe that
\[
  g(r) \defeq \Pr(G = 1 \mid R = r) = \frac {p \cdot f_{R \mid G = 1}(r)} {p
  \cdot f_{R \mid G = 0}(r) + (1 - p) \cdot f_{R \mid G = 1}(r)},
\]
where \(p \defeq \Pr(G = 1)\). Note that if
\[
  h(q) \defeq \frac {1 - p} p \cdot \frac q {1 - q},
\]
then
\[
  (h \circ g)(r) = \frac {f_{R \mid G = 1}(r)} {f_{R \mid G = 0}(r)},
\]
i.e., the likelihood ratio. As a consequence, since \(h(q)\) is monotonically
increasing, the MLRP holds in this case if and only if \(\Pr(G = 1 \mid R = r)\)
is a monotonic function of \(r\).

To understand how the MLRP connects to the proof of
Proposition~\ref{prop:simple}, let \(g^{\T {lwr}} \in \{0, 1\}\) denote the
lower ``base rate'' group and \(g^{\T {upr}} \in \{0, 1\}\) the higher base rate
group, i.e., \(g^{\T {lwr}}\) and \(g^{\T {upr}}\) are such that
\[
  \Pr(Y = 1 \mid G = g^{\T {lwr}}) \leq \Pr(Y = 1 \mid G = g^{\T {upr}}).
\]
Satisfying the MLRP implies the following two useful properties. First, the
risk-distribution of group \(G = g^{\T {lwr}}\) is ``left-shifted'' relative to
that of group \(G = g^{\T {upr}}\), i.e.,
\begin{equation}%
\label{eq:sd-simple}
  \Pr(R \leq r \mid G = g^{\T {lwr}}) \geq \Pr(R \leq r \mid G = g^{\T {upr}})
  \quad \text{for all} \ r \in [0, 1];
\end{equation}
that is, the MLRP implies that the distributions also satisfy stochastic
dominance. Secondly, the lower base rate group retains a lower base rate even
after conditioning on risk being above some threshold \(t < 1\). More
specifically,\footnote{%
  In general, the MLRP implies---and is equivalent to---uniform conditional
  stochastic dominance; here we derive Eq.~\eqref{eq:uc-mean} from uniform
  conditional stochastic dominance using the fact that, by the law of iterated
  expectations and the definition of \(R\),
  \[
    \Pr(Y = 1 \mid G = g, t \leq R) = \B E[\Pr(Y = 1 \mid X) \mid G = g, t \leq
    R] = \B E[R \mid G = g, t \leq R].
  \]
}
\begin{equation}%
\label{eq:uc-mean}
  \Pr(Y = 1 \mid G = g^{\T {lwr}}, t \leq R) \leq \Pr(Y = 1 \mid G = g^{\T {upr}}, t \leq R).
\end{equation}
See Theorems~1.C.1 and~1.C.5 in~\citet{shaked2007stochastic} for proof of these
properties.

These twin facts are enough for us to now prove Proposition~\ref{prop:simple}.

\begin{proof}[Proof of Proposition~\ref{prop:simple}]
  We proceed by proving the contrapositive: if \(t_1 \leq t_0\), then either the
  decision rate for group \(G = 0\) will be no larger than for group \(G = 1\),
  or the outcome rate will be no smaller, i.e.,
  \begin{equation}%
  \label{eq:bt}
    \Pr(D = 1 \mid G = 0) \leq \Pr(D = 1 \mid G = 1)
  \end{equation}
  or
  \begin{equation}%
  \label{eq:ot}
    \Pr(Y = 1 \mid G = 0, D = 1) \geq \Pr(Y = 1 \mid G = 1, D = 1).
  \end{equation}
  We will show that, depending on whether \(G = 0\) is the lower base rate group
  or the higher base rate group, either the benchmark test in Eq.~\eqref{eq:bt}
  or the outcome test in Eq.~\eqref{eq:ot}, respectively, will point in the
  correct direction.

  Suppose that group \(G = 0\) has the \emph{lower} base rate.
  Eq.~\eqref{eq:sd-simple} implies that
  \[
    \Pr(R \geq t_0 \mid G = 1) \geq \Pr(R \geq t_0 \mid G = 0).
  \]
  Furthermore, reducing the decision threshold from \(t_0\) to \(t_1\) can only
  increase the decision rate for group \(G = 1\). Thus, in this case, the
  decision rate for group \(G = 0\) cannot exceed that of group \(G = 1\), i.e.,
  \[
    \Pr(D = 1 \mid G = 1) = \Pr(R \geq t_1 \mid G = 1) \geq \Pr(R \geq t_0 \mid
    G = 0) = \Pr(D = 0 \mid G = 0),
  \]
  showing that Eq.~\eqref{eq:bt} holds.

  On the other hand, suppose that group \(G = 0\) has the \emph{higher} base
  rate. The outcome test looks at the base rates of the groups \emph{after}
  conditioning on receiving a positive decision. More specifically,
  \[
    \Pr(Y = 1 \mid G = g, D = 1) = \Pr(Y = 1 \mid G = g, t_g \leq R)
  \]
  by the definition of \(D\). Now, by Eq.~\eqref{eq:uc-mean}, we have that
  \[
    \Pr(Y = 1 \mid G = 0, t_0 \leq R) \geq \Pr(Y = 1 \mid G = 1, t_0 \leq R).
  \]
  Again, lowering the decision threshold from \(t_0\) to \(t_1\) only reduces
  the outcome rate for group \(G = 1\), i.e.,
  \[
    \Pr(Y = 1 \mid G = 0, D = 1) \geq \Pr(Y = 1 \mid G = 1, D = 1),
  \]
  showing that Eq.~\eqref{eq:ot} holds, and completing the proof.
\end{proof}

In proving Proposition~\ref{prop:simple}, the key insight is that, under the
MLRP, whether group \(G = 0\) is the lower or higher base rate group, either the
benchmark or the standard outcome test will correctly detect the relative
ordering of \(t_0\) and \(t_1\)---though we do not know which one. As a result,
when both tests point in the same direction, the conclusion is unambiguous.
Theorem~\ref{thm:general} below extends this argument to a much more general
setting. The chief technical obstacles there are: (1) showing that the standard
outcome test still points in the right direction for the higher base rate group,
and (2) accounting for quasi-rational decision makers and more complex risk
distributions without densities.

\section{Assessing Monotonicity}%
\label{sec:monotonicity}

The primary assumption of Proposition~\ref{prop:simple} is that \(\Pr(G = 1 \mid
R = r)\) is monotonic---which, as discussed above, is equivalent to the
group-specific risk distributions satisfying the MLRP.\@ To build intuition
about this non-parametric assumption, we consider related parametric conditions
on the group-specific risk curves. In particular, a sufficient condition for
monotonicity is that the group-specific risk curves are beta distributed with
the same variance (with possibly different means). More generally, monotonicity
holds if the risk curves are betas that cross exactly once, such as those
depicted in Figure~\ref{fig:inframarginality}. (See Appendix~\ref{app:ordering}
for more general discussion of parametric conditions that ensure monotonicity
and Figure~\ref{fig:exception} for an example where the MLRP fails.)

In our running example, equal variance roughly means that it is equally
difficult for lenders to distinguish between high- and low-risk applicants
across groups. One can imagine that an approximate version of this property
holds not only in lending, but across many domains. Indeed, if it fails to hold,
one might wonder whether decision makers are ignoring important features to mask
discriminatory intent. With redlining, for example, lenders ignored key
indicators of individual creditworthiness to justify denying loans to racial
minorities~\citep{JMLR:v24:22-1511}.

\subsection{Empirical evaluation}

We explore the extent to which monotonicity holds in practice by considering
group-specific empirical risk distributions in four domains, spanning banking,
education, and criminal justice. Specifically, we consider: (1) likelihood of
default among applicants using an online financial technology platform, using a
bevy of traditional and non-traditional variables available to the platform when
deciding whom to offer loans; (2) likelihood that law school applicants will
pass the bar exam, using their undergraduate grade-point average, LSAT score,
and other information available to schools making admissions
decisions~\citep{wightman1998lsac}; (3) risk of recidivism among defendants
awaiting court proceedings, as determined by COMPAS risk scores, which inform
judicial bail decisions~\citep{angwin2022machine, corbett2017algorithmic}; and
(4) likelihood that individuals stopped by the police are carrying contraband,
based on indicators such as the reason for the stop and the suspected offense,
which inform officer decisions to search stopped
individuals~\citep{gelman2007analysis, goel2016precinct}. We observe only a
proxy of the true outcome of interest---e.g., we see repayment outcomes only
among those who received loans, not for the entire population of applicants.
Similarly, we do not have the full suite of covariates available to decision
makers. As a result, our estimates of risk are approximate. Nonetheless, these
estimates give insight into the plausibility of the monotonicity assumption.
(See Appendix~\ref{app:data} for details on the data sources and risk estimation
methods.)

\begin{figure}
  \begin{center}
    \includegraphics{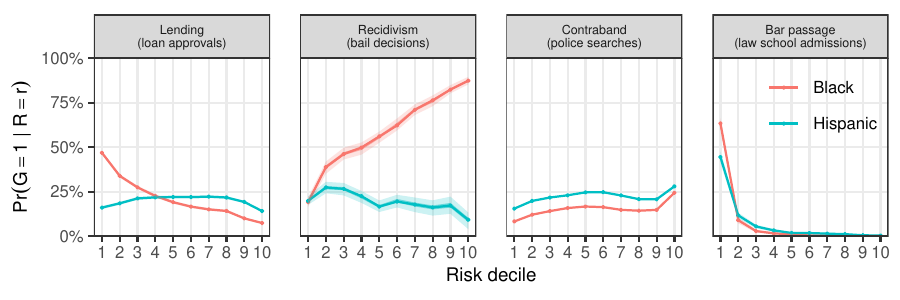}
  \end{center}
  \caption{\emph{%
    An empirical check of the monotonicity condition of
    Proposition~\ref{prop:simple} across four domains, providing evidence that
    the assumption often holds approximately in practice.
  }}%
\label{fig:montonicity}
\end{figure}

For each of these four cases, we plot, in Figure~\ref{fig:montonicity}, \(\Pr(G
= 1 \mid R = r)\) for Black vs.\ White individuals, and, separately, for
Hispanic vs.\ White individuals. (Here, ``White'' means non-Hispanic White.) We
set \(G = 1\) for the smaller group in each comparison---which corresponds to
Black or Hispanic individuals in every instance except for our policing example,
in which case White individuals are the smaller group. In every instance, we see
that the monotonicity condition holds approximately, suggesting that it is, in
practice, a relatively mild assumption. Monotonicity, however, does not hold
\emph{exactly} in these domains---nor would we expect it to in any real-world
dataset. We thus next conduct a simulation study to assess the robustness of
Proposition~\ref{prop:simple} to modest violations of monotonicity, such as
those shown in Figure~\ref{fig:montonicity}.

\subsection{Simulation study}

Starting with the empirical risk distributions in the four examples considered
above, we evaluate whether the robust outcome and standard outcome tests
correctly detect discrimination under a variety of discriminatory and
non-discriminatory scenarios. We find that across scenarios, the robust outcome
test is nearly always correct: when it indicates discrimination against a group,
that is almost always the correct inference. (Though, as expected, the test
sometimes returns an inconclusive result.) In contrast, in these simulations,
the standard outcome test often suggests discrimination against the group that
in reality received preferential treatment.

For our simulations, we compare decision and outcome rates for Black and White,
and Hispanic and White individuals under a variety of group-specific decision
thresholds \(t_g\). We sweep \(t_g\) across all percentiles of the overall risk
distribution (excluding the 0th and 100th percentiles). At each percentile, we
estimate the decision rate \(\widehat {\D {DR}}_g\) as the proportion of
individuals in group \(G = g\) whose estimated risk exceeds \(t_g\); we estimate
the outcome rate \(\widehat {\D {OR}}_g\) as the average estimated risk among
individuals in group \(G = g\) whose estimated risk exceeds \(t_g\). Based on
these quantities, we then test for discrimination using the robust and standard
outcome tests.

\begin{figure}[t!]
  \begin{center}
    \includegraphics{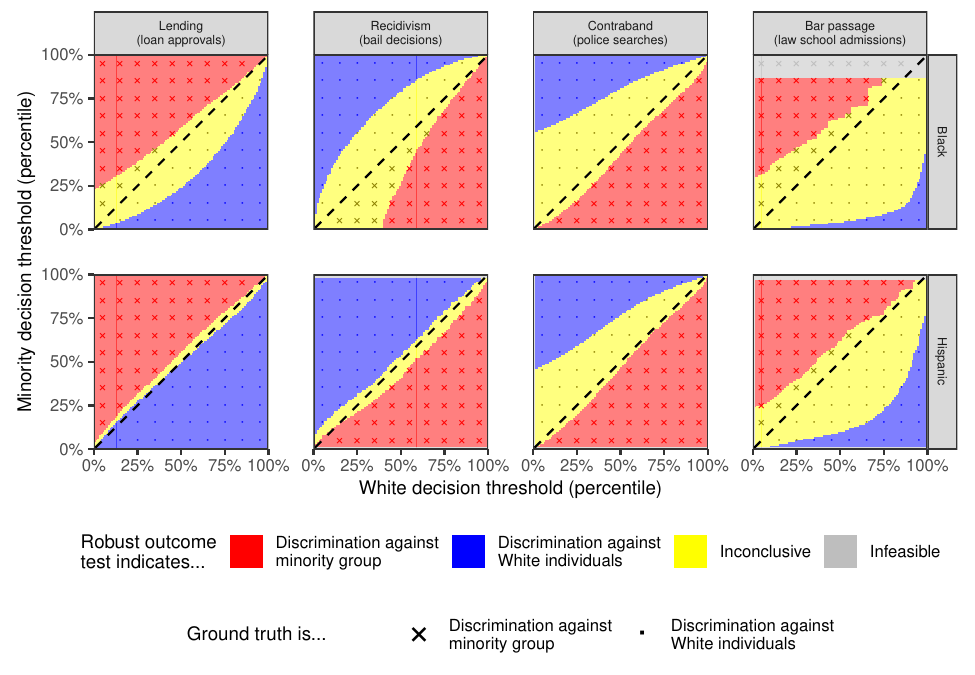}
  \end{center}
  \caption{\emph{%
    Results of a simulation study for the robust outcome test. The \(x\)-axis
    indicates the decision threshold for White individuals; the \(y\)-axis
    indicates the threshold for Black or Hispanic individuals, as appropriate.
    The upper-left and lower-right triangular regions correspond to scenarios
    where decision makers discriminate against either the minority group or
    White individuals, as indicated by the ``\,\(\times\)'' and ``\,\(\cdot\)''
    symbols, respectively; non-discriminatory scenarios are shown by the dashed
    diagonal line. Red regions indicate where the robust outcome test suggests
    discrimination against the minority group, blue regions indicate where the
    robust outcome test suggests discrimination against White individuals, and
    yellow regions indicate where the robust outcome test is inconclusive.%
  }}%
\label{fig:threshold-simulation-robust}
\end{figure}

\begin{figure}[t!]
  \begin{center}
    \includegraphics{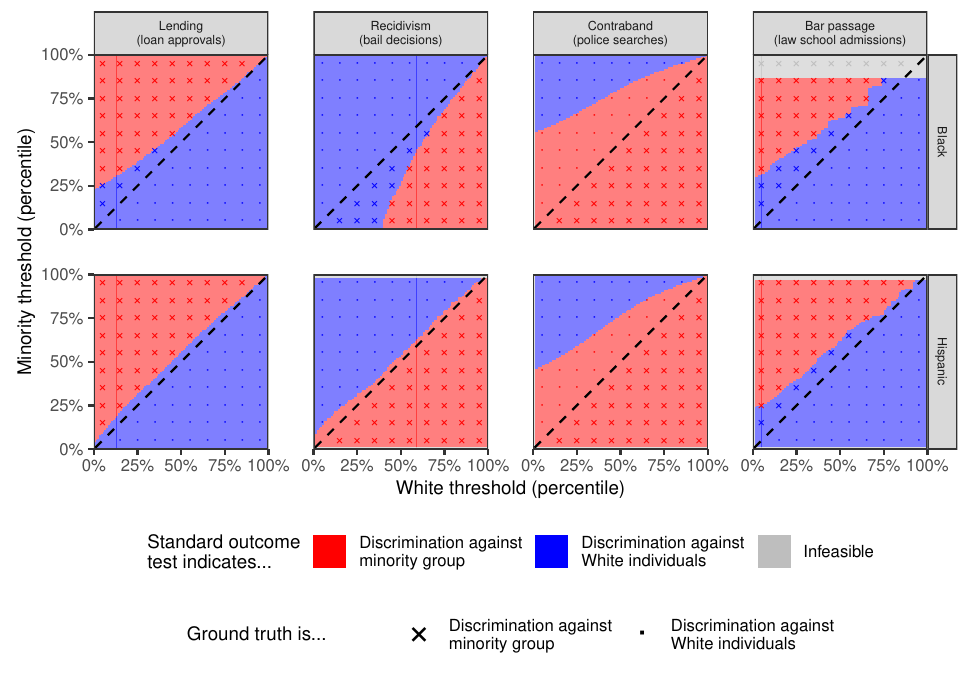}
  \end{center}
  \caption{\emph{%
    Results of a simulation study for the standard outcome test. The \(x\)-axis
    indicates the threshold for White individuals; the \(y\)-axis indicates the
    threshold for Black or Hispanic individuals, as appropriate. The upper-left
    and lower-right triangular regions correspond to scenarios where decision
    makers discriminate against either the minority group or White individuals,
    as indicated by the ``\,\(\times\)'' and ``\,\(\cdot\)'' symbols,
    respectively; non-discriminatory scenarios are shown by the dashed diagonal
    line. Red regions indicate where the standard outcome test suggests
    discrimination against the minority group, and blue regions indicate where
    the standard outcome test suggests discrimination against White individuals.%
  }}%
\label{fig:threshold-simulation-standard}
\end{figure}

The results of the simulation are shown in
Figures~\ref{fig:threshold-simulation-robust}
and~\ref{fig:threshold-simulation-standard}. As can be seen in
Figure~\ref{fig:threshold-simulation-robust}, the robust outcome test is almost
always inconclusive in the absence of discrimination---as we would hope---shown
by the yellow region covering the diagonal ``no discrimination'' line. Moreover,
in the off-diagonal regions, where the group-specific thresholds differ, the
robust outcome test frequently detects discrimination, and nearly always in the
right direction. In contrast, the standard outcome test, as shown in
Figure~\ref{fig:threshold-simulation-standard}, makes frequent errors, both
suggesting discrimination when there is none, as well as indicating
discrimination against the group that, in actuality, decision makers favored.
Thus, even in these cases where the MLRP does not hold exactly, the robust
outcome test still provides correct inferences, and, moreover, outperforms the
standard outcome test.

The extent to which the robust outcome test is able to detect
discrimination---as opposed to returning an inconclusive result---varies across
domains. In cases where base rates differ more substantially between groups, the
robust outcome test detects discrimination less frequently. In particular, the
threshold must be increased more for a lower base rate group before its outcome
rate exceeds the outcome rate of the high base rate group; and similarly, the
threshold must be increased more for a higher base rate group before its
decision rate falls below the decision rate of the low base rate group.

In our formal analysis and simulations above, we assume that decision makers are
rational within groups, making decisions based on a (potentially group-specific)
threshold. In the next section, we relax this assumption and consider
quasi-rational decision makers.

\section{General Utilities and Quasi-Rational Decision Makers}%
\label{sec:general}

The correctness of the robust outcome test holds in a more general setting than
the one presented in Section~\ref{sec:simple} that allows for both
quasi-rational decision makers, as well as more general bases for their
decisions. The general theorem and its assumptions are most naturally presented
in the language of measure theory, which we adopt here. For full details see
Appendix~\ref{app:math}.

We again imagine a population of individuals belonging to one of two groups \(G
\in \{0, 1\}\). To avoid trivialities, we assume that neither group is empty,
i.e.,
\begin{equation}%
\label{eq:assume-non-empty}
  \Pr(G = g) > 0 \qquad \text{for} \quad g \in \{0, 1\}.
\end{equation}
As above, we assume that decision makers make binary decisions \(D \in \{0,
1\}\) for each individual. We also assume that a non-zero proportion of
individuals in each group receives decision \(D = 1\), i.e.,
\begin{equation}%
\label{eq:assume-non-zero}
  \Pr(D = 1 \mid G = g) > 0 \qquad \text{for} \quad g \in \{0, 1\}.
\end{equation}
In our running lending example, loan officers base decisions on an applicant's
probability of repayment, i.e., on \(U \defeq \Pr(Y = 1 \mid X)\), where \(Y =
1\) denotes repayment of the loan and \(X\) denotes loan applicants' observable
features when lending decisions occur. Here we generalize to an arbitrary
utility \(U\), dissociating from a particular outcome \(Y\) and covariates
\(X\). For example, \(U\) could represent the lender's expected return on
lending to an applicant, rather than just the applicant's probability of
default. However, for consistency with Section~\ref{sec:simple}, we still refer
to distributions of \(U\) as \emph{risk} distributions. We assume that the
expectation of \(U\) is well-defined, i.e.,
\begin{equation}%
\label{eq:assume-finite}
  \B E[|U|] < \infty.
\end{equation}
To ensure that decisions can be compared across groups, we assume the following
analogue of the common overlap assumption~\citep{rosenbaum1983central}:
\begin{equation}%
\label{eq:assume-overlap}
  0 < \Pr(G = g \mid U) < 1 \quad \text{a.s.}
\end{equation}

In contrast to the usual setting in the outcome test literature, we do not
require there to be a single decision maker---or, more generally, a collection
of identical decision makers---making decisions based on group-specific
thresholds \(t_g\), \(g \in \{0, 1\}\). Instead, we capture the decision process
through \emph{risk-decision curves}:
\begin{equation}%
\label{eq:risk-decision}
  d_g(u) \defeq \Pr(D = 1 \mid U = u, G = g).
\end{equation}
The risk-decision curves encode the proportion of individuals who receive a
positive decision among those who belong to group \(G = g\) with utility \(U =
u\). We can ask whether one group receives positive decisions more frequently at
every level of utility, representing a double-standard. Depending on whether a
positive decision is ``desirable,'' as in lending, or ``undesirable,'' as in
policing, we understand \(d_g(u) < d_{g'}(u)\) as either discrimination or
preferential treatment.

In Section~\ref{sec:simple}, we assumed particularly simple risk-decision curves
of the following form:
\[
  d_g(u) = \bb 1(u \geq t_g),
\]
where \(\bb 1( \mathrel{\cdots} )\) denotes an indicator function. Here we
generalize beyond threshold rules to cases in which the decision makers
collectively exhibit a form of bounded rationality.

\begin{defn}[Bounded Rationality]
  We say that the risk-decision curve \(d_g(u)\) exhibits \emph{bounded
  rationality} when it is right-continuous and non-decreasing
  \(U\)-a.s.\footnote{%
    By ``\(U\)-a.s.,'' we mean that a property holds for all \(u \in \B R
    \setminus S\) where \(\Pr(U \in S) = 0\).
  }
\end{defn}

The definition of bounded rationality itself is very general, although the proof
of our main theorem requires an additional restriction defined below on the
risk-decision curves. In practice, we expect risk-decision curves to be
continuous; however, requiring only right continuity allows for the possibility
that there are thresholds where decision makers have a discontinuous increase in
their probability of choosing \(D = 1\), as, e.g., would be the case for
rational decision makers at the threshold \(t_g\).

Since it is a.s.\ bounded between zero and one, a risk-decision curve \(d_g(u)\)
exhibiting bounded rationality can be seen as the cumulative distribution
function (CDF) of a distribution \(H_g\) on the extended real numbers \(\bar {\B
R} = \B R \cup \{-\infty, \infty\}\) where
\[
  \Pr(H_g = -\infty) = \lim_{u \to -\infty} d_g(u), \quad \Pr(H_g \leq t) =
  d_g(t), \quad \text{and} \quad \Pr(H_g = \infty) = \lim_{u \to \infty} 1 -
  d_g(u).
\]
We say that \(d_g(u)\) \emph{generates} \(H_g\). Because \(d_g(u)\) is defined
only \(U\)-a.s., \(d_g(u)\) may not generate a \emph{unique} distribution
\(H_g\). But, for our purposes, the different generated distributions are
largely interchangeable, and so we will often refer to \emph{the} generated
distribution \(H_g\). (For the well-definedness of \(H_g\) and related
considerations, see Appendix~\ref{app:generation}.)

The final condition we require is that the risk-decision curves generate some
pair of distributions satisfying the MLRP. This property holds if, e.g., the
risk-decision curves are threshold rules. We now state the general version of
our main result.

\begin{thm}%
\label{thm:general}
  Suppose that the following two conditions hold:
  \begin{itemize}
    \item \(\Pr(G = 1 \mid U = u)\) is \(U\)-a.s.\ monotone,
    \item The risk-decision curves \(d_g(u)\) for \(g \in \{0, 1\}\) generate
      distributions satisfying the MLRP.
  \end{itemize}
  Under these conditions, if
  \[
    \Pr(D = 1 \mid G = 0) > \Pr(D = 1 \mid G = 1),
  \]
  and
  \[
    \B E[U \mid D = 1, G = 0] < \B E[U \mid D = 1, G = 1],
  \]
  then \(d_0(u) \geq d_1(u)\) \(U\)-a.s., where the equality is strict with
  positive probability.
\end{thm}

See Appendix~\ref{app:proof} for the proof of Theorem~\ref{thm:general}. As with
threshold decision rules, the robust outcome test holds in this more general
setting even under the modest violations of the MLRP that we see in practice;
see Appendix~\ref{app:simulation} for the results of a simulation study
analogous to the one presented above.

Theorem~\ref{thm:general} assumes that the risk-decision curves generate
distributions satisfying the MLRP.\@ This property holds in a variety of
settings, including when decisions are made according to risk thresholds \(t_g\)
as in Section~\ref{sec:simple}. It is also satisfied by logistic risk-decision
curves of the following form:
\begin{equation}
\label{eq:bounded-rationality}
  d_g(u) = \frac 1 {1 + \exp(\lambda \cdot [t_g - u])}.
\end{equation}
Here, as in the bounded rationality
literature~\citep[e.g.,][]{mckelvey1995quantal}, \(\lambda > 0\) represents the
decision makers' degree of ``rationality,'' with \(\lambda \to \infty\)
recovering threshold decision rules in the limit; and \(t_g\) represents a
``soft'' threshold at which decision makers become more likely to make decision
\(D = 1\) than \(D = 0\). Many other families of risk-decision curves satisfy
the MLRP, such as the CDFs of normal or beta distributions with the same
variance. More generally, the CDFs of log, logit, or other monotonic
transformations of normal, beta, or gamma distributions whose densities cross
once satisfy the MLRP. (See Appendix~\ref{app:ordering} for more detailed
discussion of families of distributions satisfying the MLRP, and
Figure~\ref{fig:example-policies} in the Appendix for an illustration of beta
CDF risk-decision curves satisfying the MLRP.) Intuitively, risk-decision curves
satisfy our assumption when, as in Eq.~\eqref{eq:bounded-rationality}, decision
makers are similarly sensitive to risk across groups---a property that is
strictly weaker than the ``rationality'' generally assumed in the outcome test
literature.

\section{An Application to Police Stops}%
\label{sec:ripa}

We conclude our analysis by applying the robust outcome test to data on 2.8
million police stops conducted in 2022 by 56 law enforcement agencies across
California. These data were collected as part of California's Racial Identity
and Profiling Act~\citep{RipaBoardReport2024, grossman2024reconciling}. After an
individual is stopped by the police, officers may legally conduct a search of
the individual or their vehicle if they suspect possession of contraband. Here
we use the robust outcome test to determine whether officers apply the same
standard of evidence across racial groups when deciding whom to
search.\footnote{%
  Following the outcome test literature~\citep[e.g.,][]{pierson2020large}, here
  we consider only potential discrimination in \emph{search} decisions---and not
  in, for example, stop decisions. The main advantages of focusing on search
  decisions in outcome-style analysis are: first, it is clearer what constitutes
  ``success'' (i.e., recovery of contraband); and second, outcomes are often
  reliably recorded in administrative records.
}
To do so, we compute, for each jurisdiction, the race-specific search rates and
search success rates (i.e., the proportion of searches that resulted in recovery
of contraband). If members of a group are both searched \emph{more} often and
those searches turn up contraband \emph{less} often, then the robust outcome
test suggests the group was searched at a lower standard of evidence, indicating
discrimination. (In contrast to our running lending example, where we equated
discrimination with a higher lending threshold, discrimination here corresponds
to a lower search threshold.)

\begin{figure}
  \begin{center}
    \includegraphics{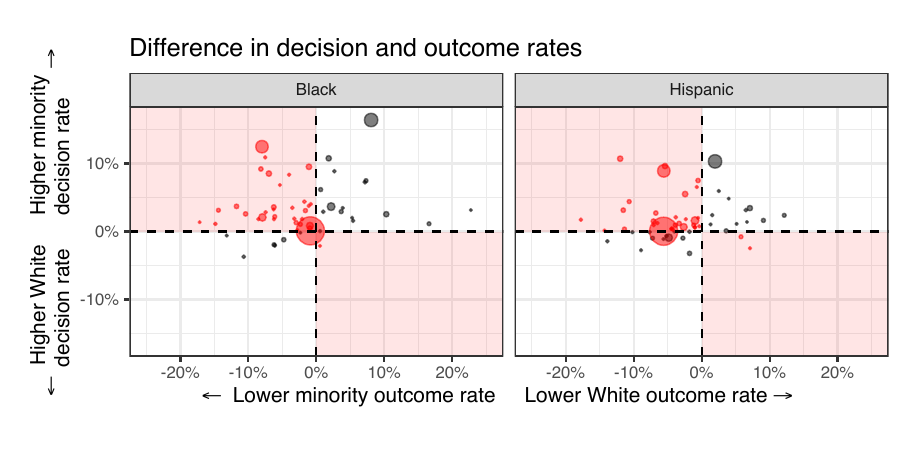}
  \end{center}
  \caption{\emph{%
    An illustration of the robust outcome test applied to 56 law enforcement
    agencies across California, with points corresponding to agencies and sized
    by the number of stops. In each panel, the robust outcome test suggests
    agencies in the upper-left quadrant discriminated against racial minorities
    when deciding whom to search, and that agencies in the lower-right quadrant
    discriminated against White individuals. The test yields inconclusive
    results for agencies in the white quadrants on the diagonal.
  }}%
\label{fig:ripa}
\end{figure}

We plot the results in Figure~\ref{fig:ripa}, with points corresponding to
agencies, sized by the number of recorded stops. Each panel compares stops of
White individuals to those of racial minorities (Black or Hispanic individuals,
respectively). In each panel, differences between group-specific search rates
are plotted on the vertical axis, and differences in search success rates on the
horizontal axis. Under the robust outcome test, the red quadrants thus indicate
racial discrimination, as those regions contain jurisdictions with both higher
search rates and lower search success rates for a given group. In particular,
the upper-left quadrants suggest discrimination against racial minorities, and
the lower-right quadrants suggest discrimination against White individuals. The
robust outcome test returns an inconclusive result for agencies in the white,
diagonal quadrants, as those correspond to both higher search rates and higher
success rates for a given group. Of the 56 agencies we consider, the robust
outcome test suggests discrimination against Black individuals by 33, and
discrimination against Hispanic individuals by 32. The test returns an
inconclusive result in nearly all of the remaining cases. The robust outcome
test thus suggests a pattern of widespread discrimination against racial
minorities in police searches across California.

The standard outcome test, in contrast, suggests White individuals were searched
at a \emph{lower} standard of evidence than Black individuals in about one-third
of agencies---corresponding to points in the right-hand quadrants---indicating
discrimination against \emph{White} people in those jurisdictions. While not
impossible, that result is at odds with an extensive analysis of police
discrimination in the literature~\citep{%
  pierson2020large, gelman2007analysis, goel2016precinct,
  grossman2024reconciling, jung2024omitted, ridgeway2007analysis,
  grogger2006testing, epp2014pulled, goel2017combatting,
  chohlas2022identifying, simoiu2017problem%
},
pointing to the statistical limitations of the standard outcome test. Due to
this lack of face validity, it is easy to dismiss results from the standard
outcome test even when it suggests more plausible findings of discrimination
against racial minorities, illustrating the value of our robust alternative.

\section{Discussion}%
\label{sec:discussion}

Our robust outcome test is a logistically straightforward and intuitively
appealing method for assessing discrimination. Applying the test requires
knowing only group-specific decision and success rates, information that is
often readily available in administrative databases. Critically, the robust
outcome test does not use individual-level covariates or decision maker
demographics, as required by other methods~\citep[e.g.,][]{%
  jung2024omitted, anwar2006alternative%
}
but which administrative records often omit.\footnote{%
  In particular, the RIPA data we analyze do not have officer demographics. They
  do have some individual-level covariates, although these covariates are
  selectively recorded, complicating statistical analyses of discrimination that
  seek to leverage that information~\citep{grossman2024reconciling}.
}
Further---and in contrast to both the benchmark and the standard outcome
tests---the robust outcome test is guaranteed to produce correct results under a
realistic assumption about the underlying risk distributions. Our empirical
analysis of police decisions further suggests that the robust outcome test is,
in practice, a more accurate barometer of bias than common alternatives.

Our theoretical and empirical results strengthen several past findings in the
outcome test literature, where decision rates were reported and consistent with
outcome rates~\citep[e.g.,][]{%
  persico2006generalising, close2007searching, ridgeway2007analysis,
  antonovics2009new, coviello2015economic, goel2017combatting, pierson2020large%
}.
In many cases, however, researchers simply apply the standard outcome test
without reporting decision rates~\citep[e.g.,][]{%
  berkovec1994race, smart1996citation, berkovec1998discrimination,
  green2009gender, anzia2011jackie, fryer2019empirical, chilton2020political%
}.
And in some instances where researchers concluded there was discrimination based
on outcome tests, the reported results of outcome and benchmark tests
diverge~\citep[e.g.,][]{pope2011s}, suggesting caution when interpreting those
findings. Our results thus highlight an important gap in the literature, and
suggest a straightforward change to improve current methodological practice.

Despite the benefits of our robust outcome test, it is important to recognize
its limitations. First, and most importantly, our proof of correctness rests on
a key monotonicity assumption. We presented empirical evidence that this
assumption holds approximately in many common cases, and we further showed that,
in practice, we obtain correct inferences even when monotonicity does not hold
exactly. But the test may yield incorrect results in settings where it is
substantially easier to make inferences about one group than another (see, e.g.,
Figure~\ref{fig:exception}). Second, computing success rates requires unbiased
outcomes. In the policing data we analyze, it seems likely that our main outcome
of interest---contraband recovery---was recorded accurately, but that may not
always be the case. Third, our robust outcome test can return inconclusive
results. In these cases, an absence of evidence of discrimination may stem
either from a lack of actual discrimination or from real discrimination that has
gone undetected. We note, though, that in our empirical analysis of police
stops, the robust outcome test produced conclusive results in the majority of
instances, revealing a pervasive pattern of discrimination. Finally, the robust
outcome test---like the standard outcome test---formally produces only a binary
determination of discrimination, not a continuous measure of the degree of
discrimination. In practice, we suspect that greater gaps in decision and
success rates point toward greater discrimination, but formally showing that
requires additional assumptions.

Recent years have brought renewed urgency to identify and ameliorate bias in
policing and beyond. We hope our work helps further this area of study, both by
providing a straightforward and statistically robust method for detecting
discrimination, and by offering a blueprint for formally studying empirical
tests of bias.

\bibliographystyle{plainnat}
\bibliography{refs}

\newpage
\appendix
\setcounter{figure}{0}
\renewcommand\thefigure{\thesection.\arabic{figure}}

\section{Mathematical Appendix}%
\label{app:math}

The proof of Theorem~\ref{thm:general} is in three parts. First, we give an
overview of important properties of risk-decision curves; in particular, we show
that risk-decision curves are well-defined and that the notion of generation,
introduced in Section~\ref{sec:general}, is also well-defined and non-vacuous.
Second, we revisit important properties of stochastic orderings, including
stochastic dominance and the MLRP.\@ Finally, we prove
Theorem~\ref{thm:general}.

\subsection{Risk-Decision Curves}%
\label{app:generation}

We begin by noting that risk-decision curves are well-defined. Recall the
definition in Eq.~\eqref{eq:risk-decision}:
\[
  d_g(u) \defeq \Pr(D = 1 \mid U = u, G = g).
\]
The risk-decision curves \(d_g(u)\) exist and are well-defined \(U\)-a.s.\ by
the Doob-Dynkin lemma because of the non-triviality assumptions in
Eqs.~\eqref{eq:assume-non-empty} and~\eqref{eq:assume-non-zero} and the overlap
assumption in Eq.~\eqref{eq:assume-overlap}.

The following lemma shows that the definition of generation is not vacuous.

\begin{lem}%
\label{lem:rdc-exists}
  Any risk-decision curve \(d_g(u)\) exhibiting bounded rationality generates
  some distribution \(H_g\).
\end{lem}

\begin{proof}
  First, we recall that \(d_g(u)\) is non-decreasing, bounded between zero and
  one, and right-continuous \(U\)-a.s. In other words, there exists a \(U\)-null
  set \(S\) such that these properties hold for all \(u \in \C R \defeq \B R
  \setminus S\). For \(k \in \left\{ 1, \ldots, 2^n \right\}\), let
  \begin{equation}%
  \label{eq:endpoint}
    \ell(k, n) \defeq \inf \; \left\{ \, u \in \C R : d_g(u) \geq k \cdot 2^{-n}
    \, \right\},
  \end{equation}
  and define \(F_n : \B R \to [0, 1]\) as follows:
  \begin{equation}%
  \label{eq:approx-cdf}
    F_n(u) \defeq \frac 1 {2^n} \sum_{k = 1}^{2^n} \bb 1(\ell(k, n) \leq u).
  \end{equation}

  The function \(F_n(u)\) approximates \(d_g(u)\). In particular, for all \(u
  \in \C R\),
  \begin{equation}
  \label{eq:bound}
    0 \leq d_g(u) - F_n(u) < 2^{-n}.
  \end{equation}
  To see why, suppose that \(u^* \in \C R\) is arbitrary and satisfies
  \[
    (k - 1) \cdot 2^{-n} \leq d_g(u^*) < k \cdot 2^{-n}.
  \]
  Then it immediately follows from Eq.~\eqref{eq:endpoint} that \(\ell(k - 1, n)
  \leq u^*\). (If \(k = 1\), then \(\ell(k - 1, n)\) is, strictly speaking, not
  defined by Eq.~\eqref{eq:endpoint}, but the proof differs only in one minor
  detail noted below.) Separately, by right continuity, there exists \(\delta >
  0\) such that for all \(u \in [u^*, u^* + \delta) \cap \C R\),
  \[
    d_g(u) < k \cdot 2^{-n}.
  \]
  By monotonicity, for all \(u \in \C R\) less than or equal to \(u^*\),
  \[
    d_g(u) \leq d_g(u^*) < k \cdot 2^{-n}.
  \]
  Therefore, for all \(u \in (-\infty, u^* + \delta) \cap \C R\),
  \[
    d_g(u) < k \cdot 2^{-n},
  \]
  and so
  \[
    u^* < u^* + \delta \leq \inf \; \left\{ \, u \in \C R : d_g(u) \geq k
    \cdot 2^{-n} \, \right\} = \ell(k, n).
  \]
  In particular, it follows that
  \[
    \left\{ u \in \C R : (k - 1) \cdot 2^{-n} \leq d_g(u) < k \cdot 2^{-n}
    \right\} = [\ell(k-1, n), \ell(k, n)) \cap \C R.
  \]
  (If \(k = 1\), simply replace the half-open interval with the open interval
  \((-\infty, \ell(1, n))\).) Since \(F_n(u)\) is exactly equal to \((k - 1)
  \cdot 2^{-n}\) on the latter set, Eq.~\eqref{eq:bound} follows.

  Now, since \(\ell(k, n) = \ell(2^m \cdot k, n + m)\), it also follows from the
  definition of \(F_n(u)\) in Eq.~\eqref{eq:approx-cdf} that for all \(n_0, n_1
  \in \B N\) and \(u \in \B R\),
  \[
    |F_{n_0}(u) - F_{n_1}(u)| < 2^{- \min(n_0, n_1)}.
  \]
  That is, the sequence \((F_n(u))_{n \in \B N}\) converges uniformly on all of
  \(\B R\). Since \(F_n(u)\) is non-decreasing, bounded between zero and one,
  and right-continuous for all \(n\), it follows that the pointwise limit has
  these properties as well. That is,
  \[
    F_{H_g}(u) \defeq \lim_{n \to \infty} F_n(u)
  \]
  is non-decreasing, bounded between zero and one, and right-continuous.
  Therefore \(F_{H_g}(u)\) is the CDF of a random variable \(H_g\).

  However, using Eq.~\eqref{eq:bound}, we also have that
  \[
    \sup \left| F_{H_g}(u) - d_g(u) \right| = \sup \lim_{n \to \infty} \left|
    F_n(u) - d_g(u) \right| \leq \sup \lim_{n \to \infty} 2^{-n} = 0,
  \]
  where the suprema are taken over \(u \in \C R\). Therefore \(F_{H_g}(u) =
  d_g(u)\) for all \(u \in \C R\). Since \(\Pr(U \in \C R) = 1\), it follows
  that \(d_g(u) = F_{H_g}(u)\) \(U\)-a.s., i.e., \(d_g(u)\) generates the
  distribution of \(H_g\).
\end{proof}

In light of Lemma~\ref{lem:rdc-exists}, we note two important points about risk
decision curves. First, a risk-decision curve can potentially generate more than
one distribution. In particular, CDFs are defined on all of \(\B R\), whereas
risk-decision curves are only defined \(U\)-a.s. Therefore, if \(U\) is not
supported on all of \(\B R\), it may be the case that there is not a unique
distribution \(H_g\) satisfying
\begin{equation}%
\label{eq:generate}
  d_g(u) = \Pr(H_g \leq u) \qquad \text{\(U\)-a.s.}
\end{equation}
In particular, \(d_g(u)\) does not determine the distribution of \(H_g\) on
regions outside the support of \(U\). The functions \(f(\,\cdot\,)\) we consider
in the proof of Theorem~\ref{thm:general} are constant on regions outside the
support of \(U\), however, meaning that the distributions of \(f(H_g)\) and
\(f(\tilde H_g)\) are nevertheless identical for distinct \(H_g\) and \(\tilde
H_g\) generated by \(d_g(u)\). Consequently, the different distributions
\(d_g(u)\) generates are, for our purposes, interchangeable. For this reason, we
will on occasion refer to ``the'' distribution generated by \(d_g(u)\), despite
the ambiguity.

Second, a generated distribution \(H_g\) is not necessarily a.s.\ finite.
Throughout, we do not assume that \emph{any} random variable is a.s.\ finite
unless explicitly mentioned, as was the case for \(D\), \(G\), and \(U\) in
Section~\ref{sec:general}.

\subsection{Stochastic Orderings}%
\label{app:ordering}

As noted in Section~\ref{sec:simple}, some hypotheses on the risk distributions
are needed in Theorem~\ref{thm:general} to rule out scenarios like the one shown
in Figure~\ref{fig:exception}. These hypotheses take the form of stochastic
ordering relations, which we review here. The simplest and most important
stochastic ordering is stochastic dominance.

\begin{defn}[Stochastic Dominance]
  Let \(F_0\) and \(F_1\) be arbitrary CDFs. We say that \(F_0\) (first-order)
  \emph{stochastically dominates} \(F_1\), written \(F_0 \succeq_1 F_1\), if for
  any \(t \in \B R\), \(F_0(t) \leq F_1(t)\). For random variables \(X\) and
  \(Y\), we write \(X \succeq_1 Y\) if \(F_X \succeq_1 F_Y\), i.e., if
  \begin{equation}
  \label{eq:sd}
    \Pr(X \leq t) \leq \Pr(Y \leq t).
  \end{equation}
  We say that a pair of random variables or distributions is
  \emph{stochastically ordered} if one stochastically dominates the other.
\end{defn}

Stochastic dominance has the useful property that increasing functions of the
dominating random variable have greater expectation.

\begin{lem}%
\label{lem:sd-exp}
  \(X \succeq_1 Y\) if and only if \(\B E[f(X)] \geq \B E[f(Y)]\) for any
  non-decreasing \(f(x)\) for which the expectations are well-defined.
\end{lem}

For proof of Lemma~\ref{lem:sd-exp}, see, e.g.,~\citet{shaked2007stochastic}.
(The statement in \citeauthor{shaked2007stochastic} implicitly assumes that
\(X\) and \(Y\) are a.s.\ finite, but the proof does not use this fact.)

To prove Theorem~\ref{thm:general}, we will need the following strengthening of
stochastic dominance.

\begin{defn}[Monotone Likelihood Ratio Property]%
\label{def:mlrp}
  Let \(F_0\) and \(F_1\) be arbitrary cumulative distribution functions. Let
  \(f_0\) and \(f_1\) be the Radon-Nikodym derivatives (i.e., densities) of
  \(F_0\) and \(F_1\) with respect to their sum measure \(\mu\).\footnote{%
    I.e., for any interval \((-\infty, t]\),
    \[
      \mu((-\infty, t]) = F_0(t) + F_1(t),
    \]
    with \(\mu\) suitably extended to all Borel sets by the Carath\'eodory
    extension theorem. Both \(F_0\) and \(F_1\) are absolutely continuous with
    respect to \(\mu\).
  }
  We say that \(F_0\) and \(F_1\) have the \emph{monotone likelihood ratio
  property} (MLRP) with \(F_0\) dominating, denoted \(F_0 \succeq_{\T {lr}}
  F_1\), if
  \[
    \frac {f_0(x)} {f_1(x)}
  \]
  is non-decreasing \(\mu\)-a.e., where, without loss of generality, we define
  the ratio to be \(\infty\) when \(f_1(x) = 0\). For random variables \(X\) and
  \(Y\), we write \(X \succeq_{\T {lr}} Y\) if \(F_X \succeq_{\T {lr}} F_Y\).
\end{defn}

The MLRP has many useful consequences and appears widely in the literature. In
addition to implying stochastic dominance (see Theorem~1.C.1
in~\citeauthor{shaked2007stochastic} and below), the MLRP holds for many
familiar parametric families used to model risk distributions. For example,
\begin{itemize}
  \item If \(X \sim \D {Beta}(\alpha_0, \beta_0)\) and \(Y \sim \D
    {Beta}(\alpha_1, \beta_1)\), then \(X \succeq_{\T {lr}} Y\) if \(\alpha_0
    \geq \alpha_1\) and \(\beta_0 \leq \beta_1\);
  \item If \(X \sim \C N(\mu_0, \sigma^2)\) and \(Y \sim \C N(\mu_1,
    \sigma^2)\), then \(X \succeq_{\T {lr}} Y\) if \(\mu_0 \geq \mu_1\);
  \item If \(X \sim \D {Gamma}(\alpha_0, \beta_0)\) and \(Y \sim \D
    {Gamma}(\alpha_1, \beta_1)\), then \(X \succeq Y\) if \(\alpha_0 \geq
    \alpha_1\) and \(\beta_0 \leq \beta_1\);
  \item If \(X \sim \D {Binom}(n, p_0)\) and \(Y \sim \D {Binom}(n, p_1)\), then
    \(X \succeq_{\T {lr}} Y\) if \(p_0 \geq p_1\).
\end{itemize}
We omit the routine verification of these facts. The MLRP is also preserved by
monotone transformations: if \(X \succeq_{\T {lr}} Y\), then \(g(X) \succeq_{\T
{lr}} g(Y)\) for monotonically increasing functions \(g(x)\); see Theorem~2.5
in~\citet{keilson1982uniform}. Thus, the MLRP also holds, e.g., with the
analogous statements for log-normal, log-gamma, and logit-normal distributions.
Like stochastic dominance, the MLRP is preserved when taking sums of random
variables, so long as the densities are log-concave; see Lemma~1.1
in~\citet{shanthikumar1986preservation}. For further examples of conditions
under which the MLRP holds, see Section~1.C of~\citet{shaked2007stochastic}. In
summary, the range of parametric families for which the MLRP holds is
considerable.

In addition, the MLRP has many different equivalent formulations, some of which
we note below.

\begin{lem}%
\label{lem:mlrp}
  Let \(X\) and \(Y\) be arbitrary random variables. Let \(Z \sim \T
  {Bernoulli}(p)\) for \(p > 0\), and let \(W \defeq Z \cdot X + (1 - Z) \cdot
  Y\). Then, the following are equivalent:
  \begin{enumerate}
    \item Monotone Likelihood Ratio Property (MLRP): The distributions of \(X\)
      and \(Y\) have the MLRP, i.e., \(X \succeq_{\T {lr}} Y\);
    \item Uniform Conditional Stochastic Dominance (UCSD): For all events \(E\)
      such that \(\Pr(X \in E) > 0\) and \(\Pr(Y \in E) > 0\),
      \[
          X \mid E \succeq_1 Y \mid E.
      \]
    \item Klinostochastic Dominance (KSD):\footnote{%
          This is the \(\B B^1\)-ordering in the notation
          of~\citet{ruschendorf1991conditional}.
      }
      \begin{equation}
      \label{eq:ksd}
        \frac {\int_{\bar {\B R}} \bb 1(x \leq t) \cdot h(x) \, dF_X(x)}
        {\int_{\bar {\B R}} h(x) \, dF_X(x)} \leq \frac {\int_{\bar {\B R}} \bb
        1(y \leq t) \cdot h(y) \, dF_Y(y)} {\int_{\bar {\B R}} h(y) \, dF_Y(y)},
      \end{equation}
      for any \(h(x)\) such that \(h(X)\) and \(h(Y)\) are a.s.\ greater than
      zero and such that
      \[
        0 < \B E[h(X)] < \infty, \qquad 0 < \B E[h(Y)] < \infty.
      \]
    \item Monotone conditional probability (MCP): The conditional probability
      \[
        g(w) \defeq \Pr(Z = 1 \mid W = w)
      \]
      is monotone \(W\)-a.s.
  \end{enumerate}
\end{lem}

\begin{proof}
  The equivalence between the MLRP and UCSD was originally given as Theorems~1.1
  and~1.3 in~\citet{whitt1980uniform}---though, see the note to Theorem~8
  in~\citet{ruschendorf1991conditional}. The equivalence between the MLRP and
  KSD was first shown in Theorems~7 and~8 in~\citet{ruschendorf1991conditional}.
  (As above, the statements in these references assume that \(X\) and \(Y\) are
  a.s.\ finite, but the proofs do not use this fact.)

  We now show the equivalence between MCP and the MLRP.\@ The proof closely
  follows the sketch of the equivalence for distributions with densities in the
  simplified proof of Proposition~\ref{prop:simple} above. Let \(p = \Pr(Z =
  1)\) and note that
  \[
    F_W = p \cdot F_{W \mid Z = 1} + (1 - p) \cdot F_{W \mid Z = 0} = p \cdot
    F_X + (1 - p) \cdot F_Y.
  \]
  In particular, it follows from the linearity of the Radon-Nikodym derivative
  and the fact that the Radon-Nikodym derivative of a measure with respect to
  itself is unity that
  \begin{equation}%
  \label{eq:decomp}
    1 - p \cdot \frac {d F_X} {d F_W} = (1 - p) \cdot \frac {d F_Y} {d F_W}.
  \end{equation}
  Next, consider the measure \(\nu : E \mapsto \Pr(W \in E \mid Z = 1)\)---i.e.,
  the distribution of \(X\)---and observe that
  \[
    p \cdot \nu[E] = \Pr(W \in E, Z = 1) = \B E[\Pr(Z = 1, W \in E \mid W)] = \B
    E[\Pr(Z = 1 \mid W) \cdot \bb 1(W \in E)].
  \]
  Here the first equality follows from the definition of conditional
  probability, the second from the law of iterated expectation, and the third
  from the \(\sigma(W)\)-measurability of \(\bb 1(W \in E)\). Now, by the
  Doob-Dynkin lemma, there exists a \(\sigma(W)\)-measurable function \(g(w)\)
  such that \(\Pr(Z = 1 \mid W) = g(W)\) a.s. Therefore, the previous expression
  equals
  \[
    \B E[g(W) \cdot \bb 1(W \in E)] = \int_{E} g(w) \, d F_W(w),
  \]
  where the last equality is the fundamental property of pushforward measures;
  see, e.g., Theorem 3.6.1 in~\citet{bogachev2007measure}. In particular, by the
  uniqueness of the Radon-Nikodym derivative,
  \[
    g(w) = \Pr(Z = 1 \mid W = w) = p \cdot \frac {d F_X} {d F_W}(w) \qquad
    \text{\(W\)-a.s.}
  \]

  Consider the function \(h(q) : [0, 1] \to [0, \infty]\) given by
  \[
    h(q) \defeq \frac {1 - p} p \cdot \frac q {1 - q},
  \]
  where, without loss of generality, we set \(h(1) = \infty\). Now,
  \[
    h \circ g = h \left( p \cdot \frac {d F_X} {d F_W} \right) = \frac {1 - p} p
    \cdot \frac {p \cdot \frac {d F_X} {d F_W}} {1 - p \cdot \frac {d F_X} {d
    F_W}} = \frac {1 - p} p \cdot \frac {p \cdot \frac {d F_X} {d F_W}} {(1 - p)
    \cdot \frac {d F_Y} {d F_W}} = \frac {\frac {d F_X} {d F_W}} {\frac {d F_Y}
    {d F_W}},
  \]
  which is the likelihood ratio. The second equality here follows from
  Eq.~\eqref{eq:decomp}. However, since \(h(q)\) is monotonically increasing,
  \(h \circ g\) is non-decreasing if and only if \(g(w)\) is. Therefore MCP
  holds if and only if the MLRP holds.
\end{proof}

Of these characterizations, the most important for proving
Theorem~\ref{thm:general} is klinostochastic dominance (KSD). For instance, we
see by taking \(h(x) = 1\) in Eq.~\eqref{eq:ksd} that KSD (and hence the MLRP)
implies stochastic dominance. KSD is closely related to the notion of tilted
distributions. For a random variable \(X\) and an \(X\)-a.s.\ non-negative
function \(h(x)\) such that \(\B E[h(X)]\) is finite and positive, we denote by
\(h \circlearrowleft X\) or \(h \circlearrowleft F_X\) the \emph{tilt} of the
distribution of \(X\) by the weight function \(h\), i.e., the probability
distribution with CDF
\begin{equation}%
\label{eq:tilt}
  F_{h \circlearrowleft X}(t) = \frac {\int_{\bar {\B R}} \bb 1(x \leq t) \cdot
  h(x) \, d F_X} {\int_{\bar {\B R}} h(x) \, d F_X} = \frac {\B E[\bb 1(X \leq
  t) \cdot h(X)]} {\B E[h(X)]}.
\end{equation}

KSD simply means that the stochastic ordering of two distributions is invariant
under tilts: by Lemma~\ref{lem:mlrp}, \(X \succeq_{\T {lr}} Y\) if and only if
\[
  h \circlearrowleft X \succeq_1 h \circlearrowleft Y
\]
for any appropriate weight function \(h(x)\) (i.e., \(h(x)\) such that \(h(X)\)
and \(h(Y)\) are a.s.\ non-negative and \(\B E[h(X)]\) and \(\B E[h(Y)]\) are
finite and positive).

\subsection{Proof of correctness of the robust outcome test}%
\label{app:proof}

KSD is the key ingredient needed to prove our main result: under appropriate
ordering assumptions on the risk distributions and risk-decision curves, the
robust outcome test is correct.

\begin{proof}
  We prove the contrapositive. That is, assume instead that
  \begin{equation}
  \label{eq:assume-contradiction}
    \Pr(d_0(U) < d_1(U)) > 0 \qquad \text{or} \quad d_0(u) = d_1(u) \quad
    \text{\(U\)-a.s.}
  \end{equation}
  Note that by our second ordering assumption, \(d_0(u)\) and \(d_1(u)\)
  generate distributions with the MLRP.\@ Since the MLRP implies stochastic
  dominance, it follows that either \(d_0(u) \leq d_1(u)\) \(U\)-a.s., or
  \(d_1(u) \leq d_0(u)\) \(U\)-a.s. By Eq.~\eqref{eq:assume-contradiction}, it
  follows that \(d_0(u) \leq d_1(u)\).

  Using the fact that \(d_0(u) \leq d_1(u)\), as in
  Proposition~\ref{prop:simple}, we will show that either:
  \begin{enumerate}
    \item The decision rate is at least as large for group \(G = 1\) as for
      group \(G = 0\), i.e.,
      \begin{equation}
      \label{eq:dr}
        \Pr(D = 1 \mid G = 0) \leq \Pr(D = 1 \mid G = 1);
      \end{equation}
    \item The outcome rate is no bigger for group \(G = 1\) than for group \(G =
      0\), i.e.,
      \begin{equation}
      \label{eq:or}
        \B E[U \mid D = 1, G = 0] \geq \B E[U \mid D = 1, G = 1].
      \end{equation}
  \end{enumerate}
  Toward that end, let \(U_0 \sim U \mid G = 0\) and \(U_1 \sim U \mid G = 1\)
  be the conditional risk distributions, and let \(d_0(u)\) and \(d_1(u)\)
  generate distributions \(H_0\) and \(H_1\) with the MLRP.\@ Then, there are
  two possibilities, which we will treat separately: either \(U_1 \succeq_{\T
  {lr}} U_0\), or \(U_0 \succeq_{\T {lr}} U_1\).

  \medskip
  \begin{case}[\(U_1 \succeq_{\T {lr}} U_0\)]
    Suppose first that \(U_1 \succeq_{\T {lr}} U_0\). Then, in particular, \(U_1
    \succeq_1 U_0\). Since \(d_0(u) \leq d_1(u)\), we have that
    \begin{align*}
      \Pr(D = 1 \mid G = 0)
        &= \B E[d_0(U) \mid G = 0] \\
        &\leq \B E[d_0(U) \mid G = 1] \\
        &\leq \B E[d_1(U) \mid G = 1] \\
        &= \Pr(D = 1 \mid G = 1),
    \end{align*}
    i.e., the decision rate is at least as high for group \(G = 1\) as for group
    \(G = 0\), which is Eq.~\eqref{eq:dr}. Here, the equalities follow from the
    law of iterated expectations and the definition of \(d_g(u)\) in
    Eq.~\eqref{eq:risk-decision}. The first inequality follows from stochastic
    dominance and the fact that \(d_0(u)\) is non-decreasing---because it
    exhibits bounded rationality---and the second inequality from the assumption
    that \(d_0(u) \leq d_1(u)\).
  \end{case}

  \medskip
  \begin{case}[\(U_0 \succeq_{\T {lr}} U_1\)]
    Next, suppose that \(U_0 \succeq_{\T {lr}} U_1\). The proof in this case is
    similar, although more delicate. We will use KSD to show that for a fixed
    risk-decision curve, the outcome rate is higher for group \(G = 0\) than
    group \(G = 1\), and then we will use KSD again in a different way to show
    that for a fixed group \(G = g\), the risk-decision curve \(d_0(u)\) results
    in a higher outcome rate than the risk-decision curve \(d_1(u)\).

    In particular, we will need to consider the distributions \(d_{g'}
    \circlearrowleft U_g\) for all \(g, g' \in \{0, 1\}\), so we begin by
    verifying that these distributions are well-defined and have finite
    expectations.

    \begin{detail}[%
      Well-definedness and finite expectation of \(d_{g'} \circlearrowleft
      U_g\)%
    ]
      Let \(d(u)\) be any risk-decision curve. Then \(\B E[d(U) \mid G = g]\) is
      the corresponding decision rate for group \(G = g\). In order for \(d
      \circlearrowleft U_g\) to be well-defined, the definition of a tilted
      distribution in Eq.~\eqref{eq:tilt} requires that
      \[
        0 < \int_{\B R} d(u) \, d F_{U_g} = \B E[d(u) \mid G = g],
      \]
      i.e., that the decision rate is positive. (For avoidance of doubt, by
      \(\int \mathrel{\cdots} \, d F_X\), we mean throughout the integral taken
      with respect to the pushforward measure \(E \mapsto \Pr(X \in E)\) on
      \(\bar {\B R}\), not the Riemann-Stieltjes integral with respect to the
      integrator \(F_X\).) When \(d(u) = d_{g'}(u)\), \(\B E[d(U) \mid G = g]\)
      is in fact positive by assumption. To see this, note that when \(g = g'\),
      \(\B E [d_{g'}(U) \mid G = g] > 0\) by the non-triviality assumption in
      Eq.~\eqref{eq:assume-non-zero}. On the other hand, when \(g \neq g'\),
      observe that the overlap assumption in Eq.~\eqref{eq:assume-overlap}
      implies that \(U_0\) and \(U_1\) are mutually absolutely continuous. Since
      \(\B E[d_{g'}(U) \mid G = g'] > 0\), there must exist some set \(E\) and
      \(\epsilon > 0\) such that \(d_{g'}(u) > \epsilon\) for all \(u \in E\)
      and \(\Pr(U \in E \mid G = g') > 0\). Therefore, in particular,
      \begin{equation}
      \label{eq:prob-exists}
        \B E[d_{g'}(U) \mid G = g] \geq \epsilon \cdot \Pr(U \in E \mid G = g) >
        0.
      \end{equation}
      Here, the first inequality follows from the fact that \(d_g(u) \geq
      \epsilon \cdot \bb 1(u \in E)\). The second follows from mutual absolute
      continuity and the fact that \(\Pr(U \in E \mid G = g') > 0\). In short,
      the tilted distributions \(d_{g'} \circlearrowleft U_g\) are well-defined
      for all \(g, g' \in \{0, 1\}\).

      Next, we verify that the expectations of the tilted distributions \(d_{g'}
      \circlearrowleft U_g\) are all finite. These expectations are the outcome
      rates of the risk-decision curve \(d_{g'}(u)\) over the conditional risk
      distribution of group \(G = g\) for all \(g, g' \in \{0, 1\}\). Let
      \(d(u)\) be an arbitrary risk-decision curve, and assume that the
      corresponding decision rate \(\B E[d(u) \mid G = g]\) is positive. Then,
      the absolute value of the outcome rate for group \(G = g\) can be bounded
      as follows:
      \begin{equation}
      \label{eq:exp-exists}
        \left| \frac {\int_{\B R} u \cdot d(u) \, d F_{U_g}} {\int_{\B R} d(u)
        \, d F_{U_g}} \right| \leq \frac {\int_{\B R} |u| \, d F_{U_g}}
        {\int_{\B R} d(u) \, d F_{U_g}} = \frac {\B E[|U| \mid G = g]} {\B
        E[d(U) \mid G = g]} \leq \frac {\B E[|U|]} {\Pr(G = g) \cdot \B E[d(U)
        \mid G = g]}.
      \end{equation}
      Here the first inequality follows from the fact that \(0 \leq d(u) \leq
      1\) \(U\)-a.s.\ combined with the triangle inequality for integrals. The
      second inequality follows from the fact that for any non-negative random
      variable \(X\) and positive probability event \(E\), \(\B E[X \mid E] \leq
      \B E[X] \mathop/ \Pr(E)\). We note that \(\B E[|U|]\) is finite by the
      assumption in Eq.~\eqref{eq:assume-finite} and \(\Pr(G = g) > 0\) by the
      assumption in Eq.~\eqref{eq:assume-non-empty}. In particular, when \(d(u)
      = d_{g'}(u)\), \(\B E[d(U) \mid G = g]\) is positive by
      Eq.~\eqref{eq:prob-exists}. Therefore the outcome rates are finite.
    \end{detail}

    \begin{detail}[%
      Comparison of outcome rates for \(d_0(u)\) and \(d_1(u)\)%
    ]
      Having verified that the decision rates are positive and the outcome rates
      are finite for groups \(G = 0\) and \(G = 1\) under both risk-decision
      curves \(d_0(u)\) and \(d_1(u)\), we move on to comparing these rates.

      Since \(U_0 \succeq_{\T {lr}} U_1\), we have by KSD in
      Lemma~\ref{lem:mlrp} that \(d_1 \circlearrowleft U_0 \succeq_1 d_1
      \circlearrowleft U_1\). It follows from Lemma~\ref{lem:sd-exp} that the
      former outcome rate is greater than or equal to the latter outcome rate.
      To see this, note that \(d_1(u)\) exhibits bounded rationality and is
      therefore non-decreasing. Additionally, the outcome rate
      \[
        \frac {\int_{\B R} u \cdot d_1(u) \, d F_{U_g}} {\int_{\B R} d_1(u) \, d
        F_{U_g}}
      \]
      is finite for \(g \in \{0, 1\}\) by Eq.~\eqref{eq:exp-exists}. Applying
      Lemma~\ref{lem:sd-exp}, we therefore have that
      \begin{equation}
      \label{eq:hit-rate}
        \B E[U \mid D = 1, G = 1] = \frac {\int_{\B R} u \cdot d_1(u) \, d
        F_{U_1}} {\int_{\B R} d_1(u) \, d F_{U_1}} \leq \frac {\int_{\B R} u
        \cdot d_1(u) \, d F_{U_0}} {\int_{\B R} d_1(u) \, d F_{U_0}}.
      \end{equation}

      Next, we show that switching the risk-decision curve from \(d_1(u)\) to
      \(d_0(u)\) can only increase the outcome rate. To see this, we proceed in
      four steps. First, note that by the definition of \(H_g\), we have that
      \[
        \frac {\int_{\B R} u \cdot d_g(u) \, d F_{U_0}(u)} {\int_{\B R} d_g(u)
        \, d F_{U_0}(u)} = \frac {\int_{\B R} u \cdot \left[ \int_{\bar {\B R}}
        \bb 1(s \leq u) \, d F_{H_g} (s) \right] \, d F_{U_0}(u)} {\int_{\B R}
        \left[ \int_{\bar {\B R}} \bb 1(s \leq u) \, d F_{H_g}(s) \right] \, d
        F_{U_0}(u)}.
      \]
      Second, by Eq.~\eqref{eq:exp-exists}, we can apply the Fubini-Tonelli
      theorem to the latter expression, yielding
      \[
        \frac {\int_{\B R} u \cdot d_g(u) \, d F_{U_0}(u)} {\int_{\B R} d_g(u)
        \, d F_{U_0}(u)} = \frac {\int_{\bar {\B R}} \left[\int_{\B R} u \cdot
        \bb 1(u \geq s) \, d F_{U_0}(u) \right] \, d F_{H_g}(s)} {\int_{\bar {\B
        R}} \left[ \int_{\B R} \bb 1(u \geq s) \, d F_{U_0}(u) \right] \, d
        F_{H_g}(s)}.
      \]
      Third, multiplying and dividing the outer integrand in the numerator on
      the right-hand side by \(\int_s^\infty 1 \, d F_{U_0}(u)\) yields that
      \begin{equation}
      \label{eq:trick}
        \frac {\int_{\B R} u \cdot d_g(u) \, d F_{U_0}(u)} {\int_{\B R} d_g(u)
        \, d F_{U_0}(u)} = \frac {\displaystyle \int_{\bar {\B R}} \left[ \frac
        {\int_{\B R} u \cdot \bb 1(u \geq s) \, d F_{U_0}(u)} {\int_{\B R} \bb
        1(u \geq s) \, d F_{U_0}(u)} \right] \cdot \left[ \textstyle \int_{\B R}
        \bb 1(u \geq s) \, d F_{U_0}(u) \right] \, d F_{H_g}(s)} {\int_{\bar {\B
        R}} \left[ \int_{\B R} \bb 1(u \geq s) \, d F_{U_0}(u) \right] \, d
        F_{H_g}(s)}.
      \end{equation}
      Fourth and finally, from the form of the integral, we see that the
      right-hand side of Eq.~\eqref{eq:trick} is an expectation. In particular,
      it is the expectation of a non-decreasing function \(f(s)\):
      \[
        f(s) \defeq \frac {\int_{\B R} u \cdot \bb 1(u \geq s) \, d F_{U_0}(u)}
        {\int_{\B R} \bb 1(u \geq s) \, d F_{U_0}(u)} =\B E[U \mid U \geq s, G =
        0].
      \]
      The expectation in Eq.~\eqref{eq:trick} is taken with respect to the
      distribution \(h \circlearrowleft H_g\), where
      \[
        h(s) \defeq \int_{\B R} \bb 1(u \geq s) \, d F_{U_0}(u) = \Pr(U \geq s
        \mid G = 0).
      \]
      Since by Eq.~\eqref{eq:trick} the expectation of \(f(s)\) with respect to
      the distribution \(h \circlearrowleft H_g\) equals
      \[
        \frac {\int_{\B R} u \cdot d_g(u) \, d F_{U_0}(u)} {\int_{\B R} d_g(u)
        \, d F_{U_0}(u)},
      \]
      the expectation is consequently finite by Eq.~\eqref{eq:exp-exists}.

      Putting these pieces together, since \(H_0 \succeq_{\T {lr}} H_1\) by
      hypothesis, it follows that \(h \circlearrowleft H_0 \succeq_1 h
      \circlearrowleft H_1\). Thus, since \(f(s)\) is non-decreasing and has
      finite expectation, it follows by Lemma~\ref{lem:sd-exp} that its
      expectation with respect to \(h \circlearrowleft H_0\) is at least as
      large as its expectation with respect to \(h \circlearrowleft H_1\). In
      other words, the expression in Eq.~\eqref{eq:trick} with \(g = 0\) is
      greater than or equal to the expression with \(g = 1\), i.e.,
      \[
        \frac {\int_{\B R} u \cdot d_1(u) \, d F_{U_0}} {\int_{\B R} d_1(u) \, d
        F_{U_0}} \leq \frac {\int_{\B R} u \cdot d_0(u) \, d F_{U_0}} {\int_{\B
        R} d_0(u) \, d F_{U_0}} = \B E[U \mid D = 1, G = 0].
      \]
      Combining this with Eq.~\eqref{eq:hit-rate} gives that the outcome rate is
      at least as small for group \(G = 1\) as group \(G = 0\), i.e.,
      \[
        \B E[U \mid D = 1, G = 1] \leq \B E[U \mid D = 1, G = 0],
      \]
      which is Eq.~\eqref{eq:or}.
    \end{detail}
  \end{case}
\end{proof}

\subsection{A Note on Sampling Error}%
\label{app:sampling}

To clarify our main conceptual points, we have focused throughout on the
infinite data regime, where decision and outcome rates are estimated with
negligible error. In practice, however, sampling error can be an important
concern.

For notational simplicity, let \(\D {DR}_g\) and \(\D {OR}_g\) denote the true
decision and outcome rates for group \(g\)---i.e., \(\Pr(D = 1 \mid G = g)\) and
\(\B E[U \mid D = 1, G = g]\). Let
\[
  n_g \defeq \sum_{i=1}^n \bb 1(G_i = g) \qquad \text{and} \quad n_{g,d} \defeq
  \sum_{i=1}^n \bb 1(G_i = g, D_i = d)
\]
denote the number of individuals in group \(G = g\) and, respectively, the
number of observations both in group \(G = g\) and receiving decision \(D = d\).
Let \(\widehat {\D {DR}}_g\) and \(\widehat {\D {OR}}_g\) for \(g \in \{0, 1\}\)
denote the empirical decision and outcome rates, i.e.,
\[
  \widehat {\D {DR}}_g \defeq \frac {n_{g,1}} {n_g}, \qquad \widehat {\D {OR}}_g
  \defeq \frac {\sum_{i=1}^n U_i \cdot \bb 1(G_i = g, D_i = 1)} {n_{g,1}}.
\]
Set
\[
    \widehat{\D {SVar}}_{U,g} \defeq \frac {\sum_{i=1}^n \bb 1(G_i = g, D_i = 1)
    \cdot (U_i - \widehat {\D {OR}}_g)} {n_g - 1}.
\]
Finally, let \(\Delta_{\D {DR}} \defeq \D {DR}_1 - \D {DR}_0\) denote the true
difference in decision rates, with \(\Delta_{\D {OR}}\), \(\hat \Delta_{\D
{DR}}\), and \(\hat \Delta_{\D {OR}}\) defined analogously.

To quantify uncertainty in \(\hat \Delta_{\D {DR}}\) and \(\hat \Delta_{\D
{OR}}\), first observe that \(\widehat {\D {DR}}_0\), \(\widehat {\D {DR}}_1\),
\(\widehat {\D {OR}}_0\), and \(\widehat {\D {OR}}_1\) are approximately
independent: dependency arises only because \(n_0\), \(n_1\), \(n_{0,1}\), and
\(n_{1,1}\) are not fixed.\footnote{%
  In particular, using the delta method and central limit theorem, it is
  straightforward to show that as \(n \to \infty\), \((\hat \Delta_{\D {DR}} -
  \Delta_{\D {DR}}) \mathop/ \widehat{\D {SErr}}_{\Delta_{\D {DR}}}\) and
  \((\hat \Delta_{\D {OR}} - \Delta_{\D {OR}}) \mathop/ \widehat{\D
  {SErr}}_{\Delta_{\D {OR}}}\) converge in distribution to independent standard
  normals.
}
As a result, we can estimate the standard errors of \(\hat \Delta_{\D {DR}}\)
and \(\hat \Delta_{\D {OR}}\) as follows:
\[
  \widehat {\D {SErr}}_{\Delta_{\D {DR}}} \defeq \sqrt {\frac {\widehat {\D
  {DR}}_0 \cdot \left( 1 - \widehat {\D {DR}}_0 \right)} {n_0} + \frac {\widehat
  {\D {DR}}_1 \cdot \left( 1 - \widehat {\D {DR}}_1 \right)} {n_1}}, \qquad
  \widehat{\D {SErr}}_{\Delta_{\D {OR}}} \defeq \sqrt {\frac {\widehat{\D
  {SVar}}_{U,0}} {n_{0,1}} + \frac {\widehat{\D {SVar}}_{U,1}} {n_{1,1}}}.
\]
From this approximate independence, it also follows immediately that the most
compact central confidence region is given by
\begin{equation}%
\label{eq:ci}
  C \left( \hat \Delta_{\D {DR}}, \hat \Delta_{\D {OR}}; \alpha \right) \defeq
  \left\{ (\Delta_{\D {DR}}, \Delta_{\D {OR}}) : \left( \frac {\Delta_{\D {DR}}
  - \hat \Delta_{\D {DR}}} {\D {SErr}_{\Delta_{\D {DR}}}} \right)^2 + \left(
  \frac {\Delta_{\D {OR}} - \hat \Delta_{\D {OR}}} {\D {SErr}_{\Delta_{\D
  {OR}}}} \right)^2 \leq x_{1 - \alpha} \right\},
\end{equation}
where \(x_{1 - \alpha}\) is the \((1 - \alpha)\)-quantile of the \(\chi^2(2)\)
distribution. Equivalently, \(C \left( \hat \Delta_{\D {DR}}, \hat \Delta_{\D
{OR}}; \alpha \right)\) corresponds to thresholding the joint density of
independent \(\C N \left( \hat \Delta_{\D {DR}}, \widehat {\D
{SErr}}_{\Delta_{\D {DR}}} \right)\) and \(\C N \left( \hat \Delta_{\D {OR}},
\widehat {\D {SErr}}_{\Delta_{\D {OR}}} \right)\) distributions above an
appropriate threshold \(t_\alpha\) to obtain the region of highest density. This
results in an elliptical confidence region for \(\Delta_{\D {DR}}\) and
\(\Delta_{\D {OR}}\) centered at \(\left( \hat \Delta_{\D {DR}}, \hat \Delta_{\D
{OR}} \right)\), with the major and minor axes determined by the estimated
standard errors of the sample difference in decision and outcome rates.
Asymptotically valid hypothesis tests can be straightforwardly obtained by, for
example, taking the maximum \(p\)-value of one-tailed tests of \(\Delta_{\D
{DR}} \leq 0\) and \(\Delta_{\D {OR}} \geq 0\).

\section{Data and Risk Estimation}%
\label{app:data}

Our empirical results are based on data drawn from five sources:
\begin{itemize}
  \item \textbf{Police Stops}: Administrative data gathered under the California
    Racial Identity and Profiling Act~\citep{RipaBoardReport2024};
  \item \textbf{Lending}: Lending data drawn from applicants to a large online
    financial technology platform;
  \item \textbf{Recidivism}: COMPAS recidivism prediction scores from Broward
    County, Florida~\citep{angwin2022machine};
  \item \textbf{Contraband}: Administrative data from the New York Police
    Department Stop, Question, and Frisk program~\citep{floyd};
  \item \textbf{Bar passage}: Law school admissions and graduation data from the
    Law School Admissions Council's Longitudinal Bar Passage
    Study~\citep{wightman1998lsac}.
\end{itemize}
Data sources, pre-processing steps, risk model specifications, exclusion
criteria, and other details are given below.

\begin{figure}
  \begin{center}
    \includegraphics{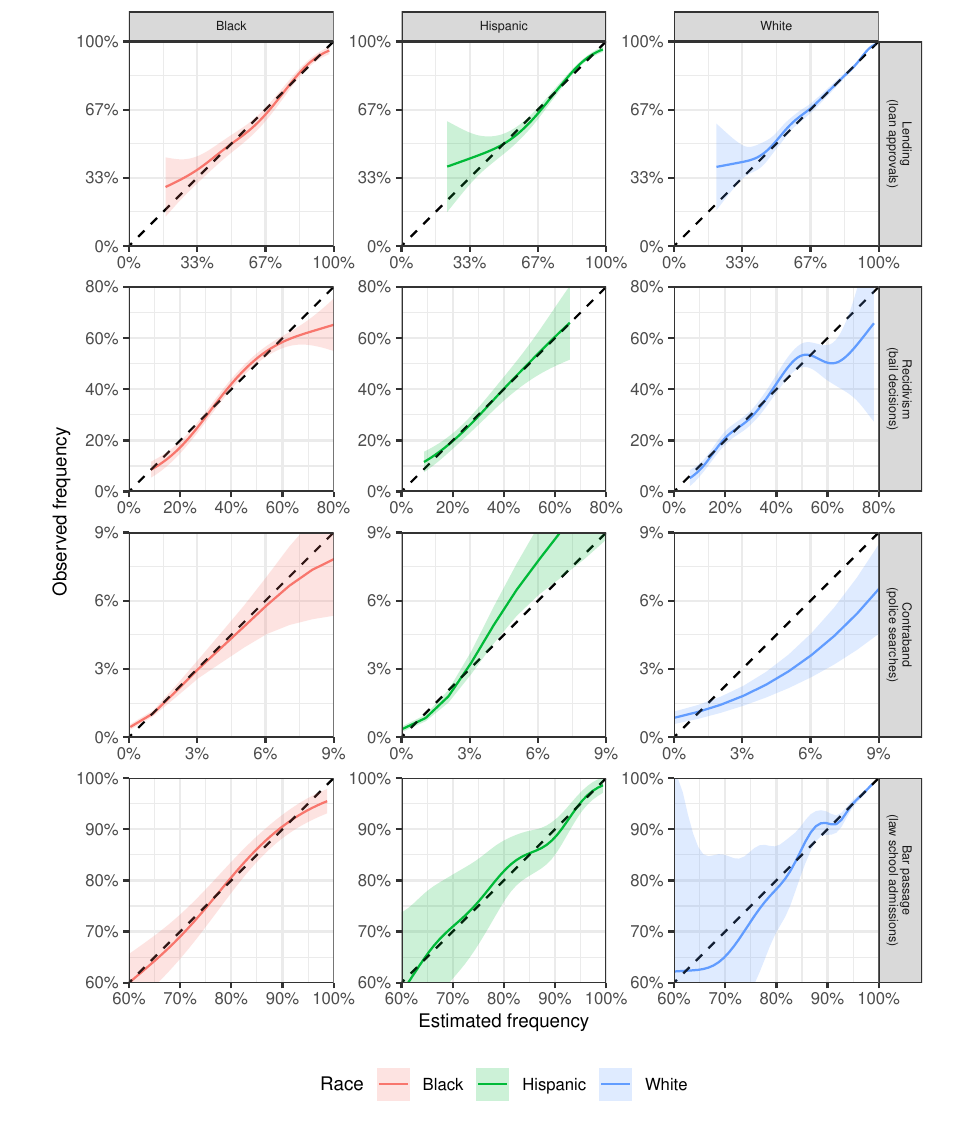}
  \end{center}
  \caption{\emph{%
    Calibration curves for the risk models used in our analysis. The \(x\)-axis
    indicates the estimated proportion of cases in which the predicted event
    (viz., loan repayment, recidivism, weapon possession, or bar passage)
    occurs; the \(y\)-axis indicates the actual proportion of cases in which the
    event occurs. The diagonal line indicates perfect calibration. The smooth
    curves represent logistic regression GAM models fit to the data using a thin
    plate regression spline basis. The shaded regions indicate 95\% confidence
    intervals.%
  }}%
\label{fig:calibration}
\end{figure}

\subsection{Police Stops}

The police stops dataset used in this analysis is drawn from data gathered in
2022 pursuant to California' Racial Identity and Profiling Act (RIPA). Under
RIPA, law enforcement officers must record a wide range of information about
stops they make of vehicles and pedestrians, including the stop circumstances
(e.g., date and time), demographics of the stopped individual (e.g., race and
age), reasons for the stop (e.g., traffic violation or matching the description
of a person of interest), outcomes of the stop (e.g., weapons or other
contraband found), actions taken during the stop (e.g., arrest or use of force),
and other information. First enacted in 2015, RIPA requirements were extended to
all law enforcement agencies in the state in 2022.

The full dataset comprises more than 4.5 million records of police stops from
536 law enforcement agencies~\citep{RipaBoardReport2024}.
Following~\citet{grossman2024reconciling}, we restrict our analysis to
non-consensual, discretionary stops not conducted in an educational context.
That is, we exclude:
\begin{enumerate}
  \item Stops undertaken because the officer had knowledge of an arrest warrant
    or some form of mandatory supervision (i.e., probation, parole, or
    post-release community supervision);
  \item Consensual encounters with motorists or pedestrians, which are recorded
    inconsistently across agencies;
  \item Stops that occurred in an educational context, or for education-related
    reasons, such as truancy, as the threshold for conducting stops of this
    nature can be lower than for other types of stops. (See,
    e.g.,~\citealt{william}.)
\end{enumerate}
To ensure that sufficient data exist to accurately estimate the decision and
outcome rates, we further restrict our analysis to agencies that conducted at
least 1,000 stops of Black, Hispanic, and White individuals in 2022. This leaves
us with a final dataset of approximately 2.8 million stops from 56 law
enforcement agencies.

\subsection{Lending}

\paragraph*{Data Source}

Our lending data come from an online financial technology platform that
facilitates lending to individuals for a variety of purposes based on a
proprietary underwriting model that utilizes both traditional and
non-traditional information to assess lending risk. Our training data consist of
approximately 300,000 borrowers who applied for loans between January 2019 and
July 2021 for whom repayment outcomes are known. While race is not recorded in
the training data, it is imputed using names and geographic information using
Bayesian Improved Surname Geocoding or ``BISG''~\citep{elliott2009using}. For
each individual, we impute the most likely race using the BISG model, and then
subset to individuals whose imputed race is Black, Hispanic, or White. We then
fit our model on the approximately 260,000 individuals remaining. To understand
the distribution of risk among loan applicants, we apply our fitted risk model
to a separate set of 130,000 individuals who applied for loans in early- to
mid-August 2021, but who did not necessarily receive a loan.

\paragraph*{Risk Model}

We fit a gradient-boosted decision tree model to predict the likelihood that an
individual will default before the end of the loan term using a bevy of
traditional and non-traditional covariates, including credit score, available at
the time of application. The model achieves an AUC of 73\% on a 10\% set of
held-out training data, and is well-calibrated, as shown in
Figure~\ref{fig:calibration}.

\subsection{Recidivism}

\paragraph*{Data Source}

The recidivism dataset used in this analysis comes from ProPublica's analysis of
COMPAS recidivism prediction scores from Broward County,
Florida~\citep{angwin2022machine}. The dataset contains COMPAS scores for
approximately 12,000 individuals who were assessed for recidivism risk pre-trial
in 2013 and 2014, as well as information on whether they were rearrested within
two years of their initial assessment. After subsetting to the collection of
individuals for whom outcomes are known, and whose race is recorded as Black,
Hispanic, or White, we are left with a final dataset of approximately 11,000
observations.

\paragraph*{Risk Model}

To predict probability of recidivism, we use raw COMPAS scores, which range
continuously from approximately -3 to 3. To obtain calibrated probabilities, we
fit a logistic regression model to predict recidivism from the interaction of
the raw COMPAS score and race. The model achieves an AUC of 68\% on the held-out
training data, which is consistent with the AUC reported for COMPAS scores in a
variety of jurisdictions~(%
  \citealt{brennan2009evaluating}, \citealt{lansing2012new},
  \citealt{dieterich2011predictive}, \citealt{farabee2010compas},
  \citealt{flores2016false}, \citealt{reich2016evidence},
  \citealt{dieterich2017compas}, \citealt{dieterich2018compas},
  \citealt{angwin2022machine}; see \citealt{equivant2019compas} for a review%
).
The model is also well-calibrated, as shown in Figure~\ref{fig:calibration}.

\subsection{Contraband}

\paragraph*{Data Source}

Under \emph{Terry v.~Ohio}, Police officers may stop and question pedestrians
for whom they have ``reasonable suspicion'' of criminal activity, and may
conduct a ``frisk'' of the individual's outer clothing if they believe the
stopped individual is in possession of a weapon~\citep{terry}. Between 2003 and
2013, the New York Police Department (NYPD) conducted over 100,000 such
\emph{Terry} stops per year. Officers recorded details of these stops, including
stop circumstances (e.g., time and location), reasons for suspicion (e.g.,
suspicious bulge or furtive movements), stop outcomes (e.g., arrest or
contraband found), and demographic information about the stopped individual. We
use records of the approximately 2.7 million stops conducted between 2008 and
2013. Our training set consists of the approximately 800,000 stops conducted
between 2008 and 2010 in which the stopped individual was frisked, and
consequently for which it is known whether or not the individual was in
possession of a weapon. We split the training data into an 80\% training set,
which we use to fit our model predicting weapon possession, and a 20\% holdout
set, which we use to evaluate the performance of our model. To understand the
overall distribution of risk among stopped individuals, we apply our fitted risk
model to the approximately 1.1 million stops conducted in 2011 and 2012.

\paragraph*{Risk Model}

We fit a gradient-boosted decision tree model to predict the likelihood that an
individual who is frisked will be found to be in possession of a weapon using 33
covariates known at the time of the stop and before the frisk decision is made.
These covariates include the time and location of the stop, the officer's reason
for making the stop, and additional circumstances related to the stop, such as
whether the individual was in proximity to a crime scene. The model achieves an
AUC of 79\% on the held-out training data. The model is reasonably
well-calibrated for risks below 9\% (approximately the 99th percentile of the
risk distribution), although it is modestly over-predictive of risk for White
individuals; see Figure~\ref{fig:calibration}.

\subsection{Admissions}

\paragraph*{Data Source}

We draw our bar passage dataset from the Law School Admissions Council's
Longitudinal Bar Passage Study~\citep{wightman1998lsac}, which studied bar
passage rates among 163 of 172 ABA-accredited law schools in the United States.
Study participants represent around 23,000 out of the approximately 40,000 law
school students who entered law school in the fall of 1991. The dataset includes
information about the law schools attended by the students, their undergraduate
GPAs, their LSAT scores, their gender and race, and whether they passed the bar
exam on their first or second attempt.

\paragraph*{Risk Model}

We fit a logistic regression model to predict the likelihood that an individual
eventually passed the bar exam using the individual's LSAT score, undergraduate
GPA, a measure of socioeconomic status, gender, and race. The model achieves an
AUC of 79\%, and is well-calibrated, as shown in Figure~\ref{fig:calibration}:
for Black and Hispanic applicants, the model is well-calibrated across a range
of estimated risks, and for White applicants, it is well-calibrated for risks
greater than roughly 90\% (approximately the 3rd percentile of risk for White
applicants).

\section{Simulation Study}%
\label{app:simulation}

We extend our simulation study to quasi-rational decision policies of the
following form:
\[
  d_{g}(r) = F_{\T {beta}}(r; t_g, \sigma^2),
\]
where \(F_{\mathrm{beta}}(x; t, \sigma^2)\) is the cumulative distribution
function of the beta function with mean \(t\) and variance \(\sigma^2\). (In the
standard beta parameterization, the distribution
\begin{equation}
\label{eq:beta-defn}
  \D {Beta}(\alpha, \beta), \qquad \text{where} \quad \alpha \defeq t \cdot
  \left(\frac {t \cdot (1 - t)} {\sigma^2} - 1 \right), \qquad \beta \defeq (1 -
  t) \cdot \left(\frac {t \cdot (1 - t)} {\sigma^2} - 1 \right),
\end{equation}
has mean \(t\) and variance \(\sigma^2\).) To model a realistic level of bounded
rationality, we set \(\sigma\) to be one half of the standard deviation of the
overall distribution of estimated risk. (In particular, we set \(\sigma\) to be
\(12\;\T {p.p.}\) for the lending dataset, \(3.5\;\T {p.p.}\) for the COMPAS
dataset, \(3.5\;\T {p.p.}\) for the contraband dataset, and \(1.8\;\T {p.p.}\)
for the bar passage dataset.) As before, we then sweep over all percentiles
\(t_g\) (excluding the 0th and 100th percentiles); see
Figure~\ref{fig:example-policies} for examples of the quasi-rational
risk-decision curves used in the simulation, where \(t_g\) is set to the
\(1/3\), \(1/2\), and \(2/3\) quantiles, respectively. At each percentile, we
estimate the decision and outcome rates for group \(G = g\) as
follows:\footnote{%
  Our notation in Eq.~\eqref{eq:sim} differs slightly from
  Appendix~\ref{app:sampling} to take advantage of the fact that the risk
  decision curves are known exactly.
}
\begin{equation}
\label{eq:sim}
  \widehat {\D {DR}}_g \defeq \frac 1 {n_g} \sum_{i=1}^n \bb 1(G = g) \cdot
  d_g(\hat R_i), \qquad \qquad \widehat{\D {OR}}_g \defeq \frac {\sum_{i=1}^n
  \bb 1(G = g) \cdot d_g(\hat R_i) \cdot \hat R_i} {\sum_{i=1}^n \bb 1(G = g)
  \cdot d_g(\hat R_i)},
\end{equation}
where \(d_g(r)\) is the risk-decision curve being tested for group \(G = g\),
\(\hat R_i\) is the estimated risk for individual \(i\), and \(n_g\) is the
number of individuals in group \(g\). (We note that this is equivalent to the
method used in the main text, differing only in that the risk-decision curve
takes values in \((0, 1)\).) We then estimate the difference in decision and
outcome rates between groups \(g_1\) and \(g_2\) as
\[
  \widehat {\D {DR}}_{g_1} - \widehat {\D {DR}}_{g_2}, \qquad \qquad \widehat
  {\D {OR}}_{g_1} - \widehat {\D {OR}}_{g_2},
\]
and test for discrimination using the robust and standard outcome tests
accordingly.

Figures~\ref{fig:beta-simulation-robust} and~\ref{fig:beta-simulation-standard}
show that broadly similar, if even more pronounced, results hold when the
risk-decision curves are quasi-rational as opposed to fully rational. In
particular, the robust outcome test rarely detects discrimination when decisions
are made in the same way for both groups, and usually detects discrimination
when there is a double standard. In particular,
Figure~\ref{fig:beta-simulation-standard} shows that in cases where the base
rates differ substantially between groups and the decision maker has limited
sensitivity to risk, the standard outcome test can essentially always ``detect''
that there is a lower threshold for the lower base rate group, regardless of the
actual decision thresholds.

\begin{figure}
  \begin{center}
    \includegraphics{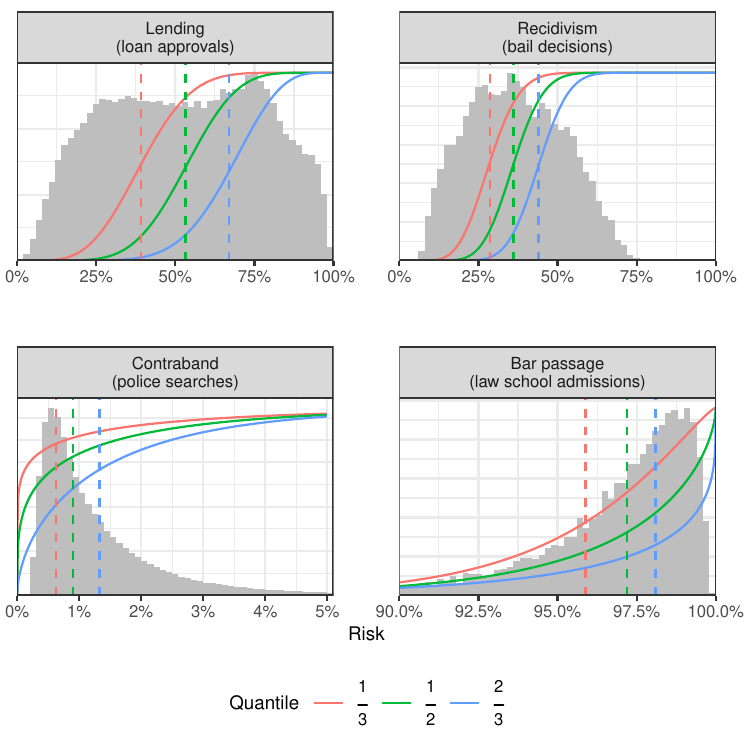}
  \end{center}
  \caption{\emph{%
    Examples of the quasi-rational risk-decision curves used in the simulation
    study. The \(x\)-axis indicates the estimated risk. The overall distribution
    of risk is shown by the grey histogram. The red, green, and blue curves show
    the quasi-rational risk-decision curves centered at the \(1/3\), \(1/2\),
    and \(2/3\) quantiles---shown by red, green, and blue dashed vertical
    lines---respectively.%
  }}%
\label{fig:example-policies}
\end{figure}

\begin{figure}
  \begin{center}
    \includegraphics{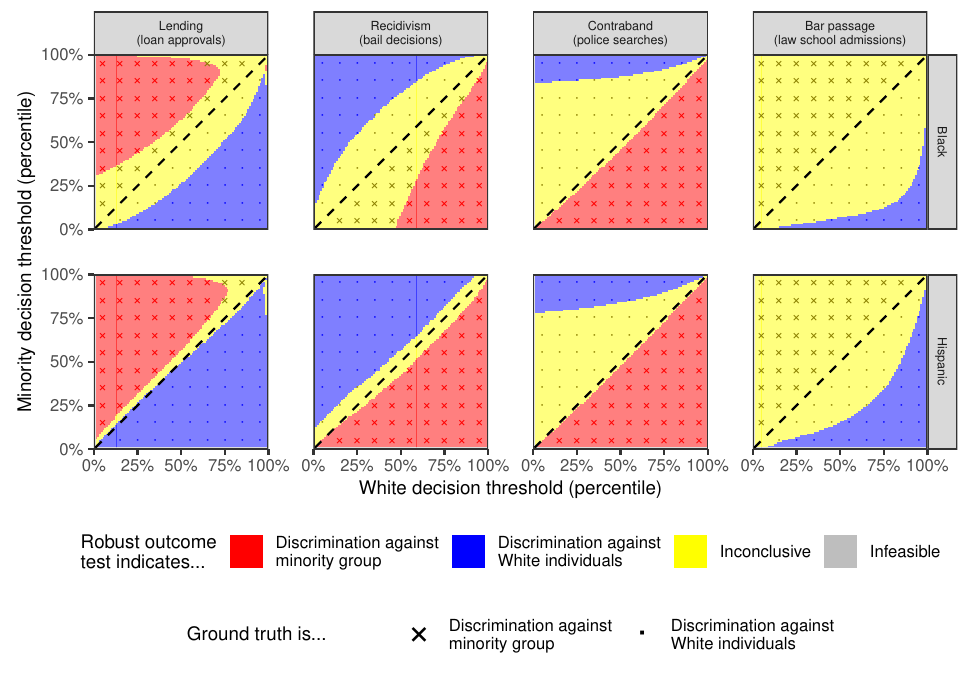}
  \end{center}
  \caption{\emph{%
    Results of the simulation study for the robust outcome test with
    quasi-rational (i.e., beta CDF) risk-decision curves. The \(x\)-axis
    indicates the ``center''---i.e., \(t\), in Eq.~\eqref{eq:beta-defn}---for
    White individuals; the \(y\)-axis indicates the center for Black or Hispanic
    individuals, as appropriate. The upper-left and lower-right triangular
    regions correspond to scenarios where decision makers discriminate against
    either the minority group or White individuals, as indicated by the
    ``\,\(\times\)'' and ``\,\(\cdot\)'' symbols, respectively;
    non-discriminatory scenarios are shown by the dashed diagonal line. Red
    regions indicate where the robust outcome test suggests discrimination
    against the minority group, blue regions indicate where the robust outcome
    test suggests discrimination against White individuals, and yellow regions
    indicate where the robust outcome test is inconclusive.%
  }}%
\label{fig:beta-simulation-robust}
\end{figure}

\begin{figure}
  \begin{center}
    \includegraphics{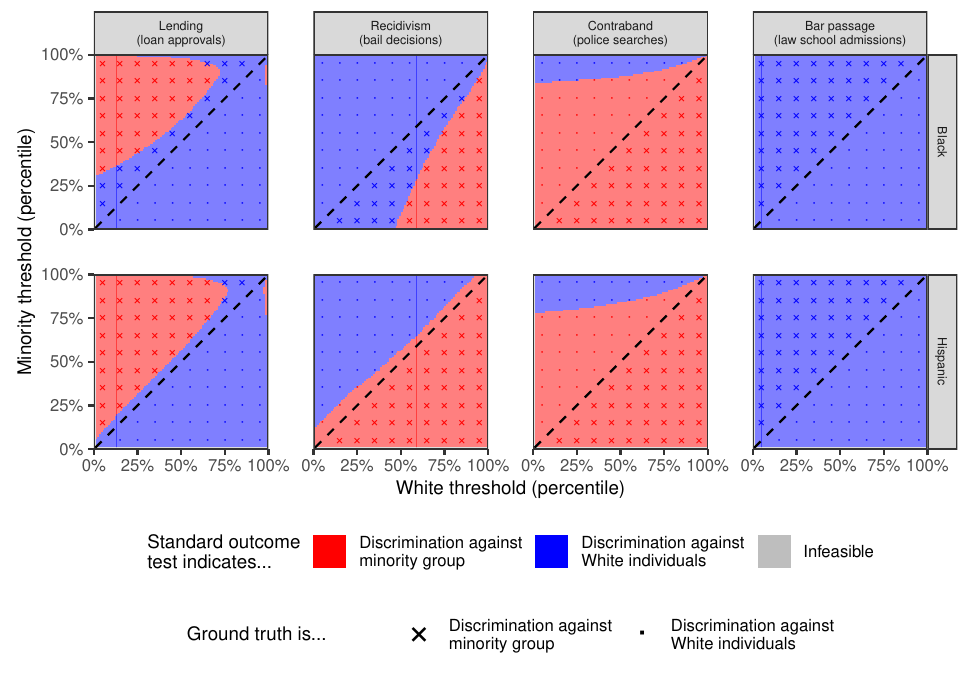}
  \end{center}
  \caption{\emph{%
    Results of the simulation study for the standard outcome test with
    quasi-rational (i.e., beta CDF) risk-decision curves. The \(x\)-axis
    indicates the ``center''---i.e., \(t\), in Eq.~\eqref{eq:beta-defn}---for
    White individuals; the \(y\)-axis indicates the center for Black or Hispanic
    individuals, as appropriate. The upper-left and lower-right triangular
    regions correspond to scenarios where decision makers discriminate against
    either the minority group or White individuals, as indicated by the
    ``\,\(\times\)'' and ``\,\(\cdot\)'' symbols, respectively;
    non-discriminatory scenarios are shown by the dashed diagonal line. Red
    regions indicate where the standard outcome test suggests discrimination
    against the minority group, and blue regions indicate where the standard
    outcome test suggests discrimination against White individuals.%
  }}%
\label{fig:beta-simulation-standard}
\end{figure}

\end{document}